\documentclass[10pt,journal]{IEEEtran}
\IEEEoverridecommandlockouts
\usepackage{cite}
\usepackage{xcolor}
\usepackage{amsmath,amssymb,amsfonts}
\usepackage{graphicx}
\usepackage{textcomp} 
\usepackage{xcolor}
\usepackage{pifont}
\usepackage{balance}    
\usepackage{xspace}
\usepackage{caption}
\usepackage{subcaption}
\usepackage{cite}
\usepackage{bm}
\usepackage{array}
\usepackage{setspace}
\usepackage[nameinlink]{cleveref}
\usepackage{mathtools}
\usepackage{amsthm}
\usepackage{multirow}
\usepackage{algorithm}
\usepackage{algpseudocode}
\usepackage{algcompatible}
\usepackage{tabularx}
\usepackage{color}
\usepackage{epstopdf}
\usepackage{endnotes}
\usepackage{framed}
\usepackage{colortbl}
\usepackage{cool} 
\usepackage{enumerate} 
\usepackage[bottom]{footmisc}
\usepackage{multirow}
\usepackage{makecell}
\usepackage{soul}
\Crefname{figure}{Fig.}{Figs.}
\usepackage{changes}
\usepackage{lipsum}
\definechangesauthor[name={Abdullatif}, color=blue]{revised/added}
\definechangesauthor[name={Abdullatif}, color=red]{Revised/Added}
\definechangesauthor[name={Abdullatif2}, color=red]{deleted}
\definechangesauthor[name={Abdullatif3}, color=red]{corrected}
\definechangesauthor[name={Abdullatif4}, color=red]{moved}

\newtheorem{assumption}{Assumption}

\allowdisplaybreaks[4]

\newcolumntype{P}[1]{>{\centering\arraybackslash}p{#1}}

\DeclareUnicodeCharacter{2212}{-}

\usepackage[justification=centering]{caption}
\allowdisplaybreaks[4]
\newcolumntype{P}[1]{>{\centering\arraybackslash}p{#1}}

\DeclareUnicodeCharacter{2212}{-}
\DeclareMathOperator*{\argmin}{arg\,min}

\usepackage{color}

\newtheorem{remark}{Remark}
\newtheorem{theorem}{Theorem}
\newtheorem{lemma}{Lemma}
\newtheorem{corollary}{Corollary}

\def\BibTeX{{\rm B\kern-.05em{\sc i\kern-.025em b}\kern-.08em
    T\kern-.1667em\lower.7ex\hbox{E}\kern-.125emX}}

\begin{document}
\title{Fair Selection of Edge Nodes to Participate in Clustered Federated Multitask Learning}

\author{Abdullatif Albaseer,~\IEEEmembership{Member,~IEEE,}
        Mohamed Abdallah,~\IEEEmembership{Senior Member,~IEEE,}
     Ala Al-Fuqaha,~\IEEEmembership{Senior Member,~IEEE,} Abegaz Mohammed,~\IEEEmembership{Member,~IEEE,}
     Aiman Erbad,~\IEEEmembership{Senior Member,~IEEE,}
      Octavia A. Dobre,~\IEEEmembership{Fellow,~IEEE}
    \thanks{Abdullatif Albaseer, Mohamed Abdallah, Ala Al-Fuqaha, Abegaz Mohammed, and~Aiman Erbad are with the Division of Information and Computing Technology, College of Science and Engineering, Doha, Qatar (e-mail:\{aalbaseer, moabdallah, aalfuqaha,mabegaz, aerbad\}@hbku.edu.qa).}
    \thanks{Octavia A. Dobre is with the Faculty of Engineering and Applied Science, Memorial University, Canada. (e-mail: odobre@mun.ca).}
\thanks{* Preliminary results in this work are presented at the IEEE GLOBECOM Conference, 2021~\cite{albaseer2021client}}
}
\maketitle

\begin{abstract}
Clustered federated Multitask learning is introduced as an efficient technique when data is unbalanced and distributed amongst clients in a non-independent and identically distributed manner.
While a similarity metric can provide client groups with specialized models according to their data distribution, this process can be time-consuming because the server needs to capture all data distribution first from all clients to perform the correct clustering.
Due to resource and time constraints at the network edge, only a fraction of devices {is} selected every round, necessitating the need for an efficient scheduling technique to address these issues.
Thus, this paper introduces a two-phased client selection and scheduling approach to improve the convergence speed while capturing all data distributions. This approach ensures correct clustering and fairness between clients by leveraging bandwidth reuse for participants spent a longer time training their models and exploiting the heterogeneity in the devices to schedule the participants according to their delay.
The server then performs the clustering depending on predetermined thresholds and stopping criteria. When a specified cluster approximates a stopping point, the server employs a greedy selection for that cluster by picking the devices with lower delay and better resources.
The convergence analysis is provided, showing the relationship between the proposed scheduling approach and the convergence rate of the specialized models to obtain convergence bounds under non-i.i.d. data distribution. We carry out extensive simulations, and the results demonstrate that the proposed algorithms reduce training time and improve the convergence speed by up to 50\% while equipping every user with a customized model tailored to its data distribution.

\end{abstract} 


\IEEEpeerreviewmaketitle

\section{Introduction}
Internet-of-Things (IoT) devices on wireless edge networks generate a vast amount of heterogeneous data that could be utilized to interpret the current behavior of the system or anticipate its future states.
 
Mobile Edge Computing (MEC) takes advantage of having edge devices and access points equipped with unrivaled capabilities to conduct complicated tasks and support intelligent services closest to these devices~\cite{chen2022joint,chen2022dynamic}. 
On MEC, \emph{Artificial Intelligence} (AI), and \emph{Machine Learning} (ML), in particular, have seen rapid advancements and have begun to deliver intelligent services that have the potential to revolutionize our lives.
Several recent studies have looked towards the use of ML techniques for IoT-edge applications, serving as an enabler for this vision~\cite{luo2020power}. Traditional ML techniques, on the other hand, require data to be outsourced and processed in a central spot, which poses a serious privacy threat, boosts the data size transferred by edge nodes, and exacerbates communication delays caused by limited resources~\cite{9352033}. 

Federated Learning (FL) is a promising decentralized ML solution that copes with these issues while preserving the data in place. 
Only model parameters (i.e., weights and biases) are shared with the server while the learning process takes place on edge devices \cite{9148111,9220780,mcmahan2016communication,9148475,10060393}. The server handles the process of developing the global model by collecting and averaging all updates performed by different edge devices. In each global round, the server regularly publishes the most recent global model for further updates. The server repeats these procedures until convergence of the global model to the optimum solution is attained.
{
Our work focuses primarily on applying FL to edge networks, so we refer to FL as Federated Edge Learning (FEEL) \cite{wang2020federated}. The main difference is that in FL, the cloud server can engage many different clients from different locations, whereas in FEEL, the training is conducted near the clients over wireless links.}

{When deploying FEEL, statistical and resource heterogeneity are considered the key challenges. 
For the statistical heterogeneity, the data amongst edge devices is distributed in a non-independent and identically distributed (non-i.i.d.) and  unbalanced fashion~\cite{mcmahan2016communication,9352033,9851847}. 
Regarding the resource heterogeneity challenge, the devices are heterogeneous, with different computation and communication capabilities. Also, bandwidth is restricted, limiting the number of devices willing to participate in a given FEEL round. This emphasizes the importance of practical and effective selection and scheduling algorithms that improve the convergence rate and produce an unbiased model that can optimally fit all incongruent data distributions, particularly in large-scale edge networks.}

{To tackle the statistical challenges, considerable research efforts have been devoted to studying and tackling this issue under different FEEL settings~\cite{zhu2021data,pang2020realizing,fallah2020personalized,sattler2020clustered,albaseer2022data,9319704,10092826}. Among all these solutions, Clustered Multi-tasks Federated Learning (CFL) \cite{sattler2020clustered,sattler2020byzantine} demonstrated exceptional performance by striking a balance between learning and cost. The geometry of the FEEL loss surface is utilized by CFL to cluster service users into groups using data distributions that can be trained at the federal level. 
It is worth mentioning that conventional FL implies that just a single global model is always being trained for all clients regardless of the discrepancies in their data distribution.
 On the other hand, Sattler \textit{et al}.~\cite{sattler2020clustered} proved that these assumptions are commonly violated in realistic applications. 
In particular, it is demonstrated that at each stationary point of the FL main objective, the similarities between the local models of different participants could be measured using cosine similarity to deduce the groupings of participants with contradictory distributions (i.e., incongruent distributions).}

{The key advantage of CFL, in to multi-task FEEL \cite{smith2017federated} which learns several models for several related tasks and personalized FEEL~\cite{zhang2020personalized} which {provides} each participant (i.e., participating client) with a customized model, {is that} {CFL} does not necessitate any change to the underlying FEEL communication protocol even though clustering occurs only after the conventional FL model reaches to {a stationary} point. CFL is interpreted as a post-processing scheme {that} can strengthen the FEEL performance by properly clustering all clients and effectively producing a customized model for each cluster. Recently, the investigations in  \cite{xie2020multi, ghosh2020efficient, briggs2020federated}  followed up on the CFL-based methodology and verified that CFL is more reliable for attaining higher accuracy than the traditional FEEL when the data is heterogeneous and non-i.i.d.}

{Several studies \cite{xu2020client, shi2020joint,yang2019scheduling,chen2022federated,9844147}  have sought to solve the challenge of participants' scheduling and resource allocation (i.e., resource heterogeneity challenges) while taking non-i.i.d. data distribution (i.e., statistical challenges). For example, the authors in \cite{xu2020client, shi2020joint,yang2019scheduling,chen2022federated,9844147} proposed several techniques to select participants during each round of federated averaging, prioritizing those with lower training latency, better communication and computation resources, or lower energy consumption, under the assumption that each client holds the same amount of data (i.e., balanced). Yet, neither of these works~\cite{xu2020client, shi2020joint,yang2019scheduling,chen2022federated,9844147}  consider the scheduling problem for the CFL technique, which requires capturing all data distributions from all clients to perform the correct clustering based on their data distributions. This means that the scheduling approaches in the literature for traditional FL are not applicable to CFL. If devices with better channels or lower latency are regularly selected to participate in the training, the resulting models will be biased since other devices not engaged in training may have different data distributions. The challenge of scheduling in CFL is a crucial issue that remains unresolved and requires further investigation.}

In response to all these remarks, this work proposes a new client scheduling framework for CFL over wireless edge networks to minimize training time and expedite the rate of convergence while equipping every group of devices with the optimal model that closely fits their data distribution. We account for the insufficient resources at the edge network (i.e., bandwidth) and the deadline constraints that the server sets to prevent a longer waiting time for updates to start a new training round. We can summarize our key contributions as follows:
\begin{itemize}
\item Propose a novel client scheduling algorithm for CFL to cope with the problem of limited resources while performing efficient clustering to tackle the non-i.i.d. and  unbalanced data distribution problems. The proposed algorithm is based on the fairness between the clients across the network in such a way that all clients have equal chances of being selected to participate in the training phase, despite their channel states and data sizes. This will enable the edge system to imbue the clients with more specialized models rather than biased models.
\item Formulate a joint optimization problem for resource allocation and scheduling aiming to reduce the training latency and improve the convergence speed, taking into account  unbalanced data distribution and non-i.i.d., device heterogeneity, and insufficient resources. Due to the NP-hardness of the problem, we propose a heuristic-based solution depending on device heterogeneity and bandwidth reuse to fairly aggregate all updates from all devices. 

\item Bound the impacts of the proposed approach on the convergence of the group's models concerning the number of scheduled clients, {and} incongruent and congruent data distribution. 

\item Perform experimental evaluation using two federated datasets, FEMNIST and CIFAR-10, under non-i.i.d and unbalanced data distribution. The results confirm that our proposed solutions effectively minimize the training latency and improve the convergence speed while attaining a satisfying performance.
\end{itemize}

The rest of the paper is structured as follows: we present the related work in Section~\ref{sec:literature}. The system, learning, computation, and communication models are introduced in Section~\ref{sysmodel}. 
Next, we formulate the minimization problem in Section ~\ref{sec:Problemformulation}. The proposed scheduling framework is introduced in Section~\ref{Sec:proposed}. The convergence analysis is given in Section~\ref{sec:conv_analysis}, where we derive the relationships between the proposed algorithm and the convergence rate.
We evaluate the proposed scheduling framework in Section~\ref{sec:experiment} where the experimental setup is outlined, and the numerical results are discussed. Finally, we summarize this paper and provide directions for future extensions in Section~\ref{conclusion}. 
\begin{figure*}[htbp]
  \centering
  \includegraphics[scale=0.3]{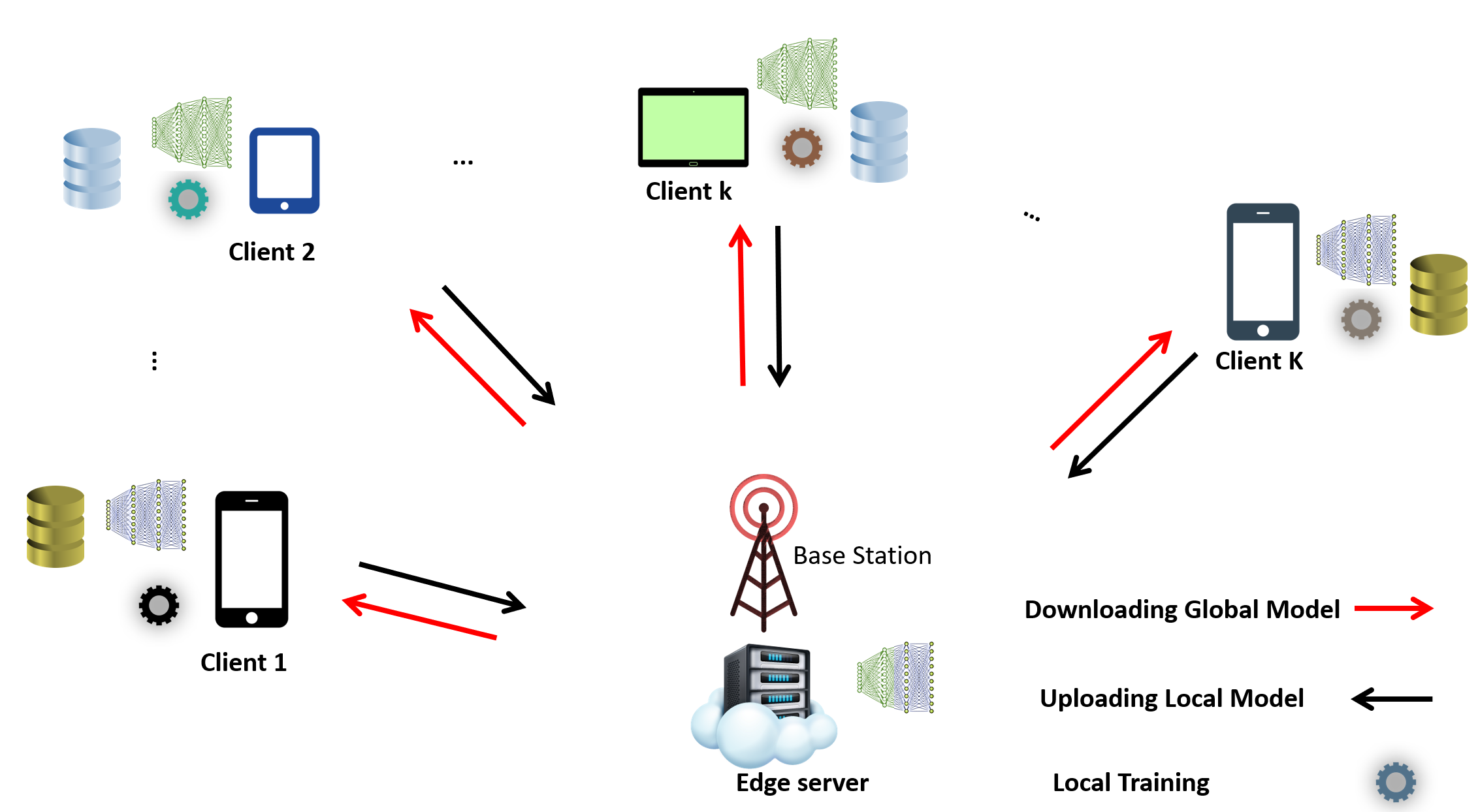} \setlength\belowcaptionskip{0cm}
  \caption{System overview where the CFL is performed at edge networks.}
  \label{WFEEL}
\end{figure*}
\section{Related Work}
\label{sec:literature}
The study of FL deployment via edge networks from various perspectives has gained a lot of interest in the literature.
A decentralized stochastic gradient descent approach was studied in~\cite{amiri2020federated} considering the context of a bandwidth limit in several channel conditions, where every client is chosen opportunistically for transmission according to its channel state. 
Earlier works considered perfect updates-upload tasks to address the delay aspects' challenges in {FEEL}. Yet, the impacts of wireless channels are ignored, especially the ability to leverage the characteristics of the channels (e.g., fading, multi-access, and broadcasting) to reduce the latency. Therefore, the broadband analogue aggregation technique is fine-tuned for probabilistic channels due to restricted radio spectrum, particularly channel capacity \cite{amiri2020federated}. 
To be more precise, prior to transmitting the models, the edge devices assess the sparsity of the gradients and then transfer them to a lower-dimensional realm constrained by the limited channel capacity. In addition, the works in \cite{chen2020convergence} investigated the convergence of the FL algorithm across wireless links, taking into account the effect of unstable connections on the global model convergence.

Recently, the work in \cite{feng2021design} introduced a solution to optimize the accuracy and cost under unreliable wireless channels. However, all analytical expressions are accounted for in the convex ML setting (not addressing the non-convexity of the deep learning algorithm). 

Focusing on scheduling policies and participants' selection, to minimize the latency during the FL training process, several client scheduling techniques were proposed in \cite{yang2019scheduling,shi2020joint,albaseer2021client,albaseer2021threshold,albaseer2021fine,9416805,9766403}. The aim is to improve the speed of the convergence rate as well as address the limitations of communication and communication resources. For example, Amiri \textit{et al}., \cite{amiri2021convergence} evaluated four alternative update quantization-based scheduling strategies. However, the authors first assumed a data symmetry among all the clients; second, they utilized stochastic gradient descent (SGD), which needs more rounds to converge on the small data. Third, data heterogeneity was not considered, so the clients with the best channels and significant L2-norm differences control the global model, resulting in a biased model. Moreover, the researchers in \cite{wu2021node} developed a new client strategy depending on the direction of the updated gradient. This strategy, however, would almost certainly result in a skewed model, particularly for a high degree of non-i.i.d. {The models with various gradient directions will not be reflected in the global model version. Lately, the work in \cite{9809926}, studied the client selection for hierarchical FL where the connections keep changing while the works in \cite{9824246,9844147} proposed approaches either to minimize the energy or to apply semi-supervised learning. Huang \textit{et al.} \cite{9766408} investigated the client selection problem, considering a volatile context scenario in such a way that some selected clients may not succeed in finishing their training and uploading tasks.  }

To conclude, although considerable works {were} devoted to optimally selecting and scheduling the participants over a wireless edge network, they only work for the traditional FL where all clients aim to train only a single model. However, CFL is entirely different from a theoretical point of view, where all incongruent data distributions must be captured in order to produce more reliable and efficient models that optimally fit the distributions of the data amongst clients. For that, a more reliable scheduling framework should be considered to enable efficient CFL systems at the network edge.

\section{System Model}
\label{sysmodel}
As \Cref{WFEEL} depicts, this work considers a set of $\mathcal{K} = \{1, \dots, K\}$ heterogeneous edge devices, $K = |\mathcal{K}|$, associated with and coordinated by an edge server via a base station (BS). Each  individual $k \in \mathcal{K}$ device owns a local dataset $\mathcal{D}_k$, with $D_k = |\mathcal{D}_k|$ number of samples, and the total number of samples amongst all devices is $D=\sum_{k=1}^{K}{D_k}$. Each local dataset, $\mathcal{D}_k$, comprises a number of sample data with input-output pairings as follows: $\{\mathbf{x}_{i,d}^{(k)},y_i^{(k)}\}^{D_k}_{i=1}$, where $\mathbf{x}_{i,d}^{(k)} \in \mathbb{R}^d$ is an input with $d$ features, and $y_i^{(k)} \in \mathbb{R}$ is the accompanying class-labeled output. Each device behaves differently based on its activities, generating different sizes of local data; thus, $\mathcal{D}_k$ is non-i.i.d. and  unbalanced. As previously stated, the connections between the devices and the server are made via an unreliable wireless link. Also, the devices themselves are resource-constrained, posing the challenge of scheduling and selecting participants in order to improve the convergence speed and minimize training time, particularly for CFL, where the main objective is to train a set of models rather than a single model as in conventional training. For the reader's convenience, the main symbols utilized in this work are listed in Table~\ref{Tab:Notationch5}.
\begin{table}[h]
\centering
\scriptsize
\caption{{List of important symbols.}}
\label{Tab:Notationch5}
\begin{tabular}{p{1cm}|p{5cm}}
 \hline
Abbrev. & Description\\ [0.5ex] 
 \hline
${K}$	& number of clients (i.e., available devices) \\
$k$	& client index $k$ where $k \in \mathcal{K}$\\
$\mathcal{D}_k$ & $k$-\textit{th} client's local data \\
$\mathcal{D}$ & Total Data samples amongst clients \\
$\{\mathbf{x}_{i,d}^{(k)},y_i^{(k)}\}$ & The input and the accompanying class-labeled output\\
${{\boldsymbol{W}}}_r$ & The parameters of the global model at $r$-\textit{th} round \\
$F_r({{\boldsymbol{W}}})$ & The loss function of the global model at $r$-\textit{th} round \\
$f_{i}$ & The local loss function for each data sample $i$ \\
${\boldsymbol{W}}_k$ & The model parameters of the $k$-\textit{th} client \\
$\mathcal{T}$ & Number of updates performed locally to update the model \\
$\eta$ & Learning rate \\
$E$ & Number of epochs \\
$b$ & Batch size \\
$T^{cmp}_k$ & Local computation time of $k$-\textit{th} client\\
$\mathcal{T}_k$  & Number of local updates \\
$T^{tot}$ & The time budget for the whole training process \\
$f_k$ & The used CPU speed at $k$-\textit{th} client\\
$\Phi_k$ & The required CPU cycles to process one local sample \\
$T_r$ & The round deadline determined by the server \\
$T^\mathrm{ trans}_k$ & The upload time (i.e., uploading latency) \\
$r_{k}^{r}$ & Transmit data rate achieved by the $k$-{th} client\\
$B$ & The bandwidth size\\
$\lambda^k_r B$ & The allocated bandwidth for client $k$ at round $r$\\
$N$ & Number of OFDMA sub-channels \\ 
$\Omega_r$ & Participant's selection set at r-\textit{th} round\\
${S}_r^k$ & The binary variable to specify whether the client is selected $1$ or not $0$ \\
$h^{k}_r$ & The uplink channel gain between the $k$-{th} client and the BS at $r$-\textit{th} round\\
$P^{r}_k$ & The $k$-\textit{th} client transmit power \\
$\xi $ & Model size \\
 $F_k(\boldsymbol{W}_r)$ & Local loss function at k-\textit{th} client\\
 $\nabla F_k({\boldsymbol{W}})$ & Local gradient at k-\textit{th} client \\
 $\varepsilon_1$, $\varepsilon_2$ & Hyper-parameters to control the clustering\\
 $I(k)$ & The data distribution of client $k$\\
\hline
\end{tabular}
\end{table}
\vspace{-0.04in}
\subsection{FEEL model}
\vspace{-0.04in}
In FEEL, the aim is to develop a collaborative global model to be used across the network. To do that, the server initiates the global model $\boldsymbol{W}^0$. Then, at every $r$-\textit{th} round, the server selects only a subset of devices ${\Omega_r}$ every round due to the limited resources (i.e., bandwidth sub-channels). Then, all participants receive a copy of the global model $\boldsymbol{W}_{r-1}$ in a multicast scheme. {Each $k$-\textit{th} selected device employs its local solver, such as SGD, to train the local model by minimizing the loss function over the number of local epochs denoted by $E$.
For example, let us assume that each device uses a mini-batch SGD; thus, at every epoch, the data is divided into batches in which the local solver performs an update on every single batch. As a result, the number of local updates performed by each participant during each global round is defined as follows:
\begin{align}
\label{local_updates}
   \mathcal{T}_k  = E~\frac{D_k}{b}, 
\end{align}
where $b$ is the batch size to determine the number of samples used for one local update.} After finishing the local training, each selected device, $k \in \Omega_r$, uploads its updated model to the edge server, which in return collects and fuses all updates to create a new global model. Regularly, the server coordinates the learning process to seamlessly find the optimal model parameters ${{\boldsymbol{W}}_r} \in \mathbb{R}^d$ that are able to learn the linked output patterns $y_i$ by repeatedly minimizing the corresponding loss function as:
\begin{equation}
  f_i({{\boldsymbol{W}}_r})=\ell(\mathbf{x}_{i,d}^{(k)},y_i^{(k)}; {{\boldsymbol{W}}_r}),
\end{equation}
 in every $r$-th global round where the local loss function is defined as:
\begin{equation}
   F_k({{\boldsymbol{W}}_r}, \mathcal{D}_k):= \frac{1}{D_k}\sum_{i \in \mathcal{D}_k} 
  f_i({{\boldsymbol{W}}_r}). 
\end{equation}
Consequently, given datasets $\mathcal{D}_1,..,\mathcal{D}_K$ amongst a set of edge devices, the global objective is to find the minimum reciprocal value of the loss function for a set of local's data $\mathcal{D}=\cup_k \mathcal{D}_k$ as follows:
\begin{equation}
\begin{aligned}
\label{eq:risk}
\min_{\boldsymbol{W}} F({\boldsymbol{W}}, D)
   =\sum_{k=1}^K\frac{D_k}{D}\underbrace{F_k({{\boldsymbol{W}}_r}, \mathcal{D}_k).}_{\text{local loss}} \quad\quad
\end{aligned}
\end{equation}
\begin{figure*}[htbp]
  \centering
  \includegraphics[scale=0.7]{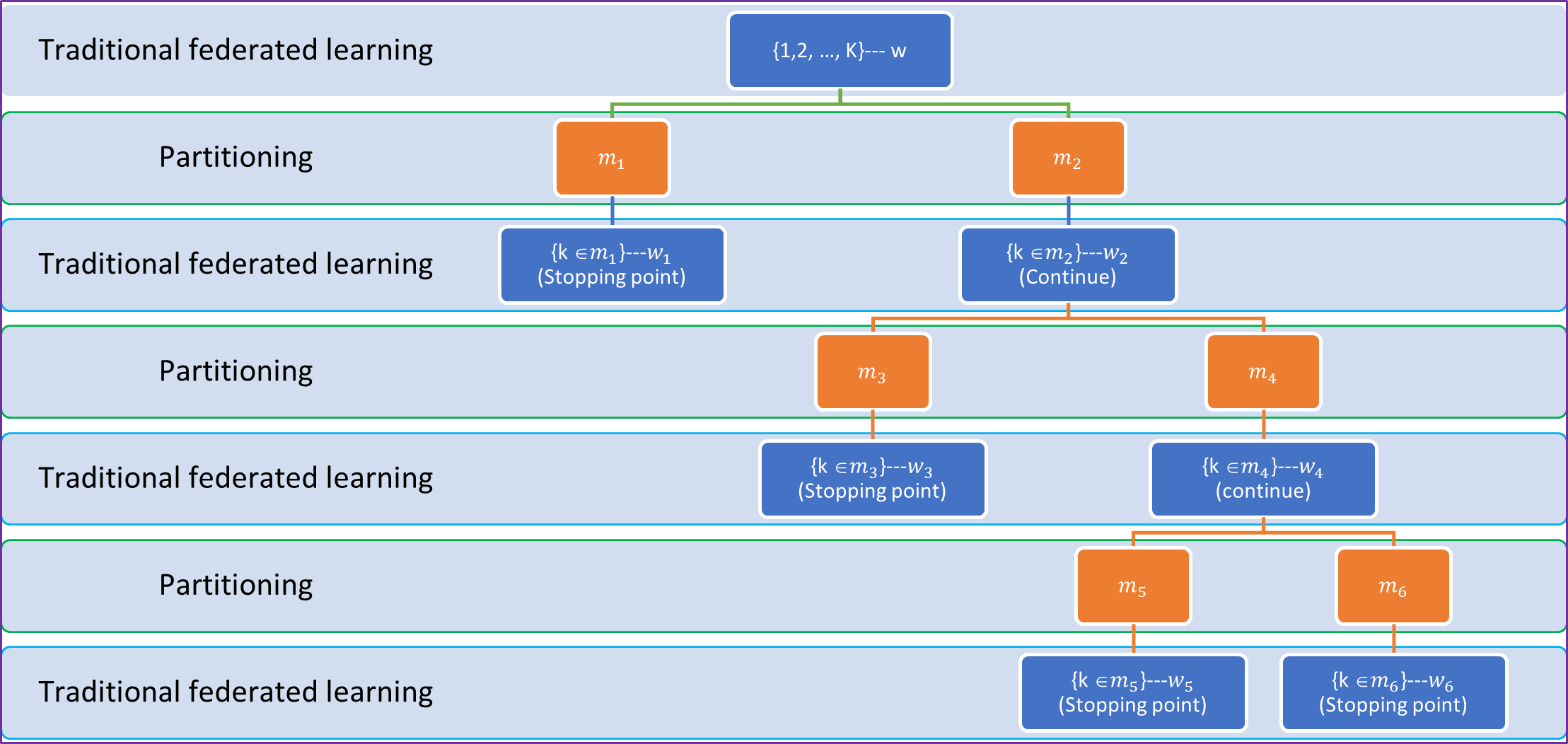}
  \caption{CFL representation as a parameter tree where the root is the traditional FL~\cite{sattler2020clustered}.}
  \label{fig:cfltree}
\end{figure*}
\vspace{-0.08in}
\subsection{Local Computation and Communication Models}
As previously indicated, the BS is unable to serve all existing devices due to a limited number of bandwidth sub-channels. Therefore, in every FEEL round, only a part of these devices, $\Omega_r$, can upload their updates. In particular, the selected set is defined as: 
\begin{equation}
  \Omega_r=\{k \mid s^k_r=1, k=1,2,\dots,K\},
\end{equation}
where $s^k_r=1$ indicates that client $k$ is on the participating set $\Omega_r$, otherwise $s^k_r=0$.
The delay amongst the selected set includes two parts, i.e., uploading delay and computing delay. For the uploading latency, all clients are assumed to hold a similar model architecture, ${{\boldsymbol{W}}_r}$, determined by the system administrator, which has a size of $\xi $. We employ the orthogonal frequency division multiple access (OFDMA) technique for the communication between the BS and the devices, where each $k$-\textit{th} participant (i.e., selected device) is given bandwidth of size $\lambda^k_r B$. To be more specific, assuming that the bandwidth is split into $N$ sub-channels depending on the model size, we can define the number of sub-channels as follows:
\begin{equation}
  N = \frac{B}{\xi},
\end{equation}
where each $1$-\textit{th} sub-channel of size $\lambda^k_r B$ is allocated for each participant. Hence, each k-\textit{th} device can achieve a data rate defined as: 
\begin{equation}
   r_k^r=\lambda^k_r B\ln{\left( 1+\frac{P^{k}_r {|h^{k}_r|^2}}{N_0} \right)},
\end{equation}
where $h^{k}_r$ is the channel gain of the link between client $k$ and the server, $P^{k}_r$ is the $k$-\textit{th} client's transmission power to upload the local model, and $N_0$ denotes background noise. Hence, the uploading delay could be estimated as:
\begin{equation}
  T^k_{trans}= \frac{\xi }{ r_k^r} \cdot \label{trans_latency} 
\end{equation}
For the computing latency, each device takes a time defined as:
\begin{equation}
  T^k_{cmp}= E\frac{\phi_k D_k}{f_k}, 
\end{equation}
to train its local model, where $f_k$ (cycles/second) is the central processing unit (CPU) frequency; and $\phi_k$ (cycles/data point) denotes CPU cycles to process one data point. Therefore, the total computing and uploading delay for each $k$-th participant at the $r$-th round is defined as:
\begin{equation}
   T^k_{tot}= T^k_{trans} + T^k_{cmp}.
\label{total_T} 
\end{equation}
{In a real FEEL scenario, the system administrator or the coordinating server determines a time constraint (i.e., a round deadline constraint) in such a way that each participating device has to finish its tasks within this time. Specifically, this time could be set to the latency of the slowest participant $k \in \Omega_r$, which is defined as:
\begin{equation} 
T_{r} = \max\{s^k_r(T^k_{trans} + T^k_{cmp})\}.
\end{equation}}
\subsection{Clustered Federated Learning (CFL)}
{Unlike traditional FL, whereby all clients train a single model collaboratively, the main goal of CFL is to cope with incongruent data distribution and provide a group of clients with a specialized model tailored to their local data distribution. More precisely, CFL~\cite{sattler2020clustered} intends to generalize the traditional FEEL assumption to a set of participating devices holding a similar data distribution so as to consolidate the aforementioned challenges.

\begin{assumption}
\label{ass:2}
\textbf{(``CFL")~\cite{sattler2020clustered}:} 
\textit{
The client population can be partitioned as $\mathcal{M}=\{c_1,\dots,c_m\}$, $\bigcup_{i=1}^M c_i = \{1,\dots,K\}$ in which every subgroup $c\in\mathcal{M}$ meets the conventional FL assumption.
}
\end{assumption}

Here, $\mathcal{M}$ is the clusters' set, and $M = |\mathcal{M}|$ is the number of clusters whereby the participating devices with similar data distribution (i.e., congruent) are grouped into a single cluster.
The FEEL system can be regarded as a representative parameter tree in CFL, as shown in \Cref{fig:cfltree} \cite{sattler2020clustered}. At the root node remains the conventional FEEL model, approaching the stationary point ${\boldsymbol{W}}^*$. 
In the subsequent layer, the client population is divided into two groups based on their cosine similarities, and every subgroup should reach a stationary point ${\boldsymbol{W}}^*_1$ and $\boldsymbol{W}^*_2$, respectively.
The recursive branching continues until no further partitions are possible. }
It is worth mentioning that the \emph{cosine similarity}, $sim$, between any two devices' updates, device $k$ and device ${k'}$, is calculated by:

\begin{footnotesize}
\begin{align}
\begin{split}
 sim_{k,{k'}}:= sim(\nabla F_k({\boldsymbol{W}}), \nabla F_k'({\boldsymbol{W}}))&:=\frac{\langle\nabla F_k({\boldsymbol{W}}), \nabla F_{k'}({\boldsymbol{W}})\rangle}{\|\nabla F_k({\boldsymbol{W}})\|\|\nabla F_{k'}({\boldsymbol{W}})\|}\\
\label{eq:cossim}
&=\begin{cases}
1 & \text{ if }I(k)=I(k')\\
-1 & \text{ if }I(k)\neq I(k'),
\end{cases}
\end{split}
\end{align}
\end{footnotesize}
\noindent where both $I(k)$ and $I(k')$ denote the data distributions of $k$ and $k'$, respectively. Accordingly, the correct bi-partitioning is obtained by: 
\begin{equation}
  m_1=\{k| sim_{k,0}=1\},~m_2=\{k| sim_{k,0}=-1\}.
\end{equation}
As in \cite{sattler2020clustered}, the separation is only executed if the following two conditions are fulfilled:
\begin{align}
\label{eq:servernorm}
0 \leq \left|\left|\sum_{k=1,\dots,K}\frac{D_k}{D}\nabla_{\boldsymbol{W}} F_k({\boldsymbol{W}}^*)\right|\right|<\varepsilon_1.
\end{align}  
We can note that \eqref{eq:servernorm} retains that the achieved solution is approaching the stationary point of the conventional FL objective. In contrast, the participants are away from the stationary point of their local loss if it is aligned with the following condition:
\begin{align}
\label{eq:clientnorm}
\max_{k=1, \dots, K} \|\nabla_{\boldsymbol{W}} F_k({\boldsymbol{W}}^*)\|>\varepsilon_2 > 0,
\end{align}
where $\varepsilon_1$ and $\varepsilon_2$ {control} hyperparameters to manage the clustering task.
\begin{remark}
The split will not be performed if all clients have the same data distribution. The reason for this is that both of the aforementioned conditions cannot be fulfilled for the same data distribution. As a result, we revert to the traditional FL with a single model.
  \end{remark}
\section{Problem Formulation}
\label{sec:Problemformulation}
As in \cite{xie2020multi, ghosh2020efficient, briggs2020federated}, all participants in CFL can be either included when the training process starts, or just a random subset is selected every round. Due to limited resources, the former assumption is impracticable for deploying CFL at the network edge (i.e., the number of sub-channels with respect to the model size). On the other hand, the latter will not catch all incongruent data distributions amongst participants, posing the challenge of effectively developing a client selection and scheduling strategy while maintaining edge wireless network restrictions and attaining required performance (i.e., efficient specialized models for all clusters). To be more specific, to properly theorize the setting of conventional FL, it is essential to partition clients of incongruent data distributions as early as possible to expedite the convergence rate and minimize associated costs in light of bandwidth as well as data and device heterogeneity. Hence, improving the convergence speed and reducing the total training latency for all resulting specialized models (i.e., the $M$ group models and one conventional FL model) depends on how we optimally select and schedule the client to compute and upload their updates and how to allocate the resources for the uploading task.  
Let $R$ be the number of global rounds completed throughout the training time budget of $T_{tot}$. The {goal} is to find the optimal selecting and scheduling strategy that obtains the optimal set of $M$ models, $\boldsymbol{W^*} = \{{{\boldsymbol{W}}_m}| m=1, 2, \dots, M\}$ that minimizes the global loss function for every group within $T_{tot}$. Mathematically speaking, the main objective is to find the optimal model parameters for each group as {follows}:  
\begin{equation}
\setlength\abovedisplayskip{3pt}
\setlength\belowdisplayskip{3pt}
      {{{\boldsymbol{W}}_m}} \triangleq \underset{{{\boldsymbol{W}_m}}\in \{ {{\boldsymbol{W}}}_r^{{\Omega_{r}}}:r=1,2,\dots,R\}}{\text{arg\,min}} F({{\boldsymbol{W}}_m}).  
      \label{def-tildew}
\end{equation}
Hence, the minimization problem for all clusters can be posed as follows:  
\begin{align}
\setlength\abovedisplayskip{3pt}
\setlength\belowdisplayskip{3pt}
      \underset{{{\boldsymbol{W^*}}}, R, { \Omega_r}, {T}_{[R]}}{\text{min}} \quad & \sum_{m=1}^M F({{{\boldsymbol{W}}_m}}) \tag{P1}\label{P1}\\
      \text{s.t.} \qquad \;\; \quad & F({{{\boldsymbol{W}}_m}}) - F({{{\boldsymbol{W}}^*_m}}) \le \epsilon,   \quad   \forall m \in \mathcal{M},   \tag{P1.0}\label{P10} \\  
      & \sum_{r=1}^R T_{r}(\Omega_r) \leq T_{tot}, \quad   \forall r \in [R],   \tag{P1.1}\label{P11}\\
      &s^k_r T^k_{tot}\leq T_r(\Omega_r), \quad   \forall r \in [R], \forall k \in \mathcal{K}, \tag{P1.2}\label{P12} \\
      & T_{r} = \max\{s^k_r(T^k_{trans} + T^k_{cmp})\}, \quad \forall r, \forall k,   \tag{P1.3}\label{P13} \\  
      & {|\Omega_r|} \le N, \quad   \forall r \in [R], \tag{P1.4}\label{P14}\\
      & s^k_r \in \{0,1\} \tag{P1.5}\label{P15},
\end{align}
where ${ \Omega_r} = [\Omega_1, \Omega_2, \dots, \Omega_R]$ denotes the selected scheduling sets during the training rounds{,} and ${T}_{[R]} = [T_1, T_2, \dots, T_R]$ is the maximum latency of every round (i.e., the deadline).  
Constraint ~\eqref{P10}   is intended to guarantee that every subgroup model converges to the optimal model {. We} can see that this constraint ensures fairness across the groups where each group has optimal model parameters. The constraint \eqref{P11}   indicates that the total training time during the FEEL process does not exceed the assigned time budget. Furthermore, constraint~\eqref{P12} is related to the deadline constraint as well as the training and upload latency for each cooperating client. Constraint \eqref{P13} specifies the deadline at every round.  
  In constraint~\eqref{P14}, the selected clients should not surpass the system bandwidth sub-channels. Last, constraint \eqref{P15} is a binary variable to determine whether the client is selected $s^k_r = 1$ or not $s^k_r = 0$. We can notice that \ref{P1} is a
Mixed-Integer Nonlinear Programming (MINLP) and NP-hard
{the} problem as the impacts of $R$ and ${\Omega}_{[R]}$ on the weight vector of each cluster model, i.e., $F({{{\boldsymbol{W}}_m}})$, should first be found.  
It is worth mentioning that it is difficult to obtain an explicit expression for $F({{{\boldsymbol{W}}_m}})$ when considering $R$ and $\Omega {[R]}$. Indeed, finding the optimal ${{{\boldsymbol{W}}_m}}$ depends on the correctness of the clustering and the corresponding congruent data distribution.   Thus,   $F({{{\boldsymbol{W}}_m}})$ is solved iteratively, as we see later in Section \ref{Sec:proposed}. Furthermore, because of the fluctuation in the computation delay,   $T^k_{cmp}$, and wireless channel conditions $h^{k}_r$   throughout the rounds $R$, obtaining the optimal client scheduling technique is complicated and might even be non-stationary. In practice, the CFL mechanism is recursive, relying on the scheduled clients to assist in capturing all incongruent data distributions and provide all clients with the congruent data distribution with the optimal model parameters. Hence, the problem in \ref{P1} is solved iteratively later on as follows.
First, to address the difficulties encountered by requiring to schedule all clients in each round to prevent biased models, we propose an efficient scheduling algorithm that accounts for the challenges mentioned above. Then, we extend the CFL algorithm, taking into account computing and communication resources and the deadline constraints. Last, we analytically find the relationship between the selected participants, the round deadline, and the convergence rate for the congruent and incongruent data distributions.  
\setlength{\textfloatsep}{1pt}

\section{Proposed Solution}
\label{Sec:proposed}
This section presents our proposed approach including both scheduling mechanisms and the related resource allocation. We take advantage of employing bandwidth reuse and devices' heterogeneity to increase the number of participating clients to capture all data distributions. 
\begin{figure}[t]
\centering
  \includegraphics[width=0.85\linewidth]{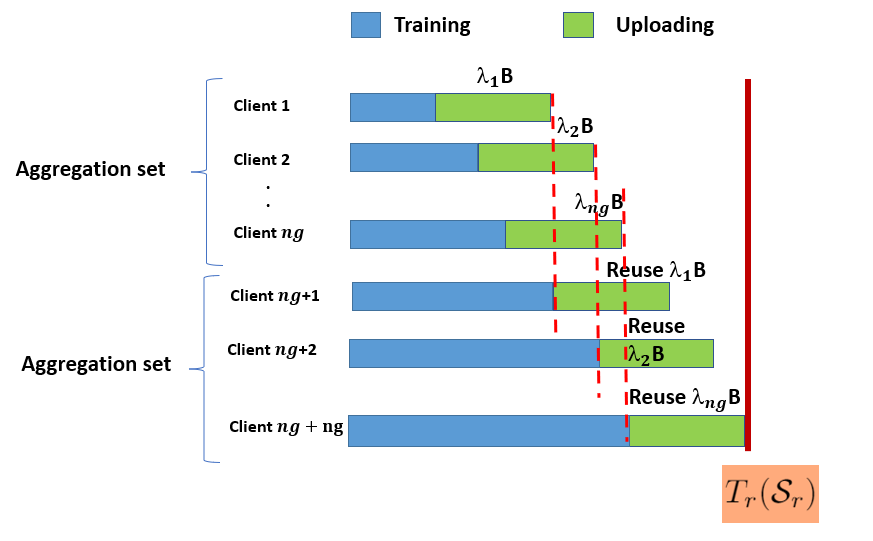}
\caption{Aggregation sets and bandwidth reuse every $r$-\textit{th} round.}
\label{fig:schamatec}
\end{figure}
In this algorithm, the scheduling and selection procedure includes two phases. In the first phase, the server copes equally with all available clients to perform the correct clustering at the early stages of the training. Thus, all existing clients are selected to carry out the global model updates until a specific cluster reaches the stopping point (i.e., all cluster members have the same data distribution). In the second phase, the server uses the greedy approach for each cluster, reaching the stopping point where the clients with less latency are only selected to carry out further model updates. Let us assume that $ \Omega_r$ includes all clients for the clusters not reaching the stopping point. The selected clients for other clusters can be added as follows:
\begin{equation}
\begin{aligned}
\label{eq:newparticipant}
  \Omega_{grdy} = {} & \{i \gets \argmin_{i \in {m}} \{T_i^{tot} \forall i \in m \\ & | \forall m| \max_{k \in m} \|\nabla_{\boldsymbol{W}} F_k({\boldsymbol{W}}^*)\|< \varepsilon_2\}\}
\end{aligned}
\end{equation}
\begin{equation}
      \Omega_r =   \Omega_r  \cup  \Omega_{grdy}.
\end{equation}
{We can see that in \eqref{eq:newparticipant}, the greedy selection is applied only to the clusters that reach the stopping point of splitting, where the stopping point is used to decide to use greedy selection for any cluster and added it to the selection set as in (18).}

{Specifically, the server initially collects prior information from all clients, such as the size of data, local processing capacity, and channel condition; subsequently, the server assesses the anticipated delays of all clients to arrange model uploading depending on their completion time. 
The server sorts the clients in ascending order according to their estimated latencies to collect the updates effectively. At this point, the number of aggregations performed by the server to gather all updates per round could be expressed as follows:
\begin{equation}
\label{eq:ng}
   ng = \frac{|{ \Omega_r}|}{N}.
\end{equation}
As shown in \Cref{fig:schamatec}, the server can exploit the heterogeneity of the devices to enable a more efficient training process where the devices with a small data size can increase their local CPU speed and transmit power to finish their computing and uploading tasks. At the same time, the other clients having large data sizes can adjust their local CPU speed and the transmit power based on the completion time of the previous aggregation set. We can infer that as the first participant in the aggregation set $j$ completes its uploading task, the first participant in the aggregation set $j+1$ will reuse the same bandwidth frequency, i.e., sub-channel, to upload its update. This procedure is {subsequently repeated} for participants in all aggregation sets.} Formally, we can define the aggregation set as follows:

\begin{footnotesize}
\begin{align}
\label{eq:gset}
     \mathcal{G}  = &\{\{1+(N(j-1)), \dots, N+(N(j-1))\}, j=1, 2, \dots, ng\}.
\end{align}
\end{footnotesize}
To find the resources for both the computation and communication tasks, each client needs to adjust its local CPU speed and the transmit power, relying on the pre-determined deadline.  
\begin{algorithm}[]
\footnotesize
\caption{: {The detailed steps of our proposed approach for efficient CFL.}}
\label{alg:proposedscheduling}
\begin{algorithmic}[1]
\STATEx \textbf{In}: $K$ , $w_0$, $\varepsilon_1, \varepsilon_2 > 0$, $E$, $b$
\STATEx \textbf{Out}: $M$ specialized models, one conventional FL model
\STATEx \textbf{Init:} Start with $\mathcal{M}=\{ \{ 1,\dots,K\} \}$ as initial clustering,  set $w_k \leftarrow w_0$ $\forall k$, and $r=1$
\WHILE {all $M$ does not reach the stopping point}
\begin{center}
\textbf{\textit{// Server Side - Before local updates start}}
\end{center}
\IF{$M > 1$ } 
\STATE Server \textbf{chooses} \\ \quad\quad $\{m \in  \mathcal{M} | \max_{k \in m} \|\nabla_{\boldsymbol{W}} F_k({\boldsymbol{W}}^*)\|<\varepsilon_2 \geq 0 \}$ 
\ELSE 
\STATE  \textbf{schedule} all participants willing to take part in the training process
\ENDIF
\STATE Server \textbf{measures} the approximate delay of $\mathcal{K}$ and sort them in ascending 
\STATE Server \textbf{organizes} the aggregation set using \eqref{eq:ng} and \eqref{eq:gset}
\STATE Server \textbf{sends} $w_{r-1}$ in multi-cast manner 
\STATEx
\begin{center}
\textbf{\textit{// Participants Simultaneously perform}}
\end{center}
\FOR {$k$ = 1 to $|{ \Omega_r}|$}
\STATE \textbf{get} $w_{r-1}$ 
\STATE \textbf{Carry out} local training $ w_k = \boldsymbol  w_{r-1} - \eta_r \sum_{t=1}^{\mathbf{\mathcal{T}}} \nabla F_k(\boldsymbol  w_k(t))$ 
\ENDFOR
\STATE Server \textbf{collects} models as in \eqref{eq:gset}
\STATEx
\begin{center}
\textbf{\textit{// Server Performs the Post-processing Steps}}
\end{center}
\STATE $\mathcal{M}_{tmp} \leftarrow \mathcal{M}$
\STATE Server \textbf{finds} $F({w^r}) =\sum_{k=1}^{K}\frac{D_k}{D}{ F_k({w^r})}$ and ${w^r} =\sum_{k=1}^{K}\frac{D_k}{D}{ {w^r}}$

\FOR {$c$ $\in$ $\mathcal{M}$}
\STATE \textbullet~ Server \textbf{obtains} $\Delta W_m \leftarrow \frac{1}{|c|}\sum_{k\in c}\Delta w_k$ 
\IF {$\|\Delta w_c\|<\varepsilon_1$  \textbf{and} $\max_{k\in c}\|\Delta w_k\|>\varepsilon_2$}
\STATE \textbullet~ $ sim_{k,{k'}}\leftarrow\frac{\langle \Delta w_k,\Delta w_{k'}\rangle}{\|\Delta w_k\|\|\Delta w_{k'}\|}$
\STATE \textbullet~ $c_1,c_2\leftarrow \arg\min_{c_1\cup c_2 = c}(\max_{k\in c_1, k'\in c_2}  sim_{k,k'})$
\STATE \textbullet~ $ sim_{cross}^{max} \leftarrow \max_{k\in c_1, {k'}\in c_2} sim_{k,k'}$
\STATE \textbullet~ $\gamma_k:=\frac{\|\nabla F_{I(k)}(\boldsymbol{W}^*)-\nabla F_k(\boldsymbol{W}^*)\|}{\|\nabla F_{I(k)}(\boldsymbol{W}^*)\|}$
\IF{$max(\gamma_k) < \sqrt{\frac{1- sim_{cross}^{max}}{2}}$}
\STATE \textbullet~ $\mathcal{M}_{tmp}\leftarrow(\mathcal{M}_{tmp}\setminus c) \cup c_1\cup c_2$
\ENDIF
\ENDIF
\ENDFOR

\STATE \textbullet~ $\mathcal{M} \leftarrow \mathcal{M}_{tmp}$
\STATE $r = r+1$
\ENDWHILE
\STATE \textbf{Server \textit{Return}} $M$ specialized models, and one conventional FL model.
\end{algorithmic}
\end{algorithm}
{It is worth highlighting that the server is anticipated to have considerably more computing capabilities than the clients; therefore, the computation cost on the server side is negligible.
We can see that there is a trade-off in the scheduling problem between the number of participants chosen in a given round and the entire number of rounds, which could be handled by $T_r$. Assuming we started with a short $T_{r}$, the number of engaging participants in $ \Omega_r$ can decrease, which is unsuitable for CFL leading to slowly partitioning the clients into congruent groups. This brings out unwanted communication costs due to increasing the number of training rounds needed to converge. Hence, we propose to initiate a long $T_{r}$ at the beginning of the training; then, this time is reduced as the correct clustering is performed where only the best clients that require less latency will be selected from the cluster reaching the stopping point as follows:
\begin{align}
\label{eq:stopping}
\max_{k\in m} \|\nabla_{\boldsymbol{W}} F_k({\boldsymbol{W}}^*)\| < \varepsilon_2.
\end{align}
The server uses \eqref{eq:stopping} to check if all clients within cluster $m$ have congruent data distribution. If so, the optimal solution $\boldsymbol{W}^*$ is returned, and the CFL is stopped for that group. Otherwise, the server partitions and clusters the clients using  \eqref{eq:cossim}, \eqref{eq:servernorm}, and \eqref{eq:clientnorm} at every round. In this context, the similarities between cooperating participants $k$ and $k'$ within a group could be bounded as follows: 
\begin{align}
\label{eq:within}
 sim_{within}^{min} = \min_{\underset{I(k)=I({k'})}{k,{k'}}} sim(\nabla_w r_k(\boldsymbol{W}^*), \nabla_w r_{k'}(\boldsymbol{W}^*)),
\end{align}
At the same time, we can bound the similarities between participants $k$ and ${k'}$ belonging to two different clusters as follows:
\begin{align}
\label{eq:cross}
 sim_{cross}^{max}=\max_{k\in c_1^*,{k'}\in c_2^*} sim(\nabla_{\boldsymbol{W}} r_k(\boldsymbol{W}^*), \nabla_{\boldsymbol{W}} r_{k'}(\boldsymbol{W}^*)).
\end{align} 
It is worth highlighting that \eqref{eq:within} and \eqref{eq:cross} are employed to measure the separation gap which can be expressed as follow:
\begin{align}
g( sim):= sim_{intra}^{min}- sim_{cross}^{max}. 
\end{align} 
All steps of our proposed framework, including the client scheduling algorithm and the CFL training algorithm, are listed in Algorithm \ref{alg:proposedscheduling} where the algorithm has input parameters $K$, $\varepsilon_1, \varepsilon_2 > 0$, $E$, and $b$. Algorithm 1 starts with an initialization step by clustering all clients in one group. Then, the server aggregates prior information such as the data size, channel state, and computational capabilities from all existing clients (line 2). In lines 3-4, the algorithm checks whether the clustering was performed before or not. If achieved, the server checks the stopping condition \eqref{eq:stopping}, and selects the best clients having less latency to take part in the training from that group, while all clients are selected for others. Otherwise, the server determines all clients to participate in the FEEL round. } In Algorithm.~\ref{alg:proposedscheduling}, we observe that whenever a given cluster approaches the stationary point of its model (lines 2–7), the server shifts to employing a greedy selection algorithm in which the clients with less training time and better resources are selected to conduct subsequent updates. This is due to the fact that clients with similar data distribution (i.e., congruent data distribution) behave likewise. For each cluster not reaching a stationary point, the server estimates all clients' latency (line 8) to schedule uploading the local models' updates and find the aggregation set (line 9) due to the limited bandwidth. In line 10, the server broadcasts the latest updated models to the selected clients. In lines (11–14), all chosen clients perform conventional FEEL updates and return them to the coordinating, which collects (line 15) all updates based on the aggregation sets. In lines 16-29, the server runs the clustering algorithm to cluster the participants according to their local data distribution if the model of the last cluster reaches the stationary points. Otherwise, the server will keep using the traditional FEEL algorithm. All these steps are repeated until all models reach the stationary points and all incongruent data distributions are captured as in Algorithm.~\ref{alg:proposedscheduling}. {It is worth mentioning that since the server depends on the uploaded models, clustering the participants isn't affected by the mobility of the devices as long as they are within the server's coverage area.} 
\section{Analysis of the Convergence Under the Proposed Algorithm}
\label{sec:conv_analysis}
This section analyzes the relationship between our proposed approach and the convergence rate of the specialized models that optimally fit $M$ incongruent data distributions.
To begin with, let us define ${\boldsymbol{W}}^{*}_m = \mathrm{arg} \min_{\boldsymbol{W}} F({\boldsymbol{W}}_m^*)$ as an optimal parameters of the specialized model for cluster $m$ that is related to the minimum loss $F(\boldsymbol{W}_m^{*})$ across all group's devices. We also define  $\boldsymbol{W}^{*}_k = \mathrm{arg} \min_{\boldsymbol{W}_k} F({\boldsymbol{W}}_{k}^{*})$ as an optimal local model parameters at client $k$. Accordingly, the optimal gap between the global of a cluster $m$ and local loss can be expressed as:
\begin{align}\label{GammaDef}
g(F_m) \triangleq F({\boldsymbol{W}}_m^*) - \frac{1}{{|\Omega^m_r|}} \sum\limits_{k \in [m]}^{{|\Omega^m_r|}} F^*_k,    
\end{align}
where $g(F_m) \geq 0$ indicates that the cluster's model $\boldsymbol{W_m}$ does not reach the optimal solution yet. For a given cluster when reaching the optimal solution, $g(F_m)$ approaches zero. 

To begin with, we apply the commonly used assumptions in FEEL \cite{wang2019adaptive,tran2019federated,amiri2021convergence}  listed as follow: 
\begin{assumption}\label{AssumpSmoothLoss}
Local loss functions $F_1, \dots, F_K$ amongst clients are all $\beta$-smooth; that is, $\forall \boldsymbol{W'}, {\boldsymbol{W}} \in \mathbb{R}^d$, 
\begin{equation}\label{ConvLSmoothCondit}
\begin{aligned}
F_k(\boldsymbol{W'}) - F_k({\boldsymbol{W}}) \le \langle \boldsymbol{W'} - {\boldsymbol{W}} , \nabla F_k ({\boldsymbol{W}}) \rangle + \frac{\beta}{2} \left\| \boldsymbol{W'} - {\boldsymbol{W}} \right\|^2_2, \\
~~~~~~~~~~~~~ \quad \forall k \in \mathcal{K}.
\end{aligned}
\end{equation}
\end{assumption}

\begin{assumption}\label{AssumpStrongConvexLoss}
Local loss functions $F_1, \dots, F_K$ amongst clients are all $\alpha$-strongly convex; that is, $\forall \boldsymbol{W'}, {\boldsymbol{W}} \in \mathbb{R}^d$, 
\begin{equation}
\begin{aligned}
\label{ConvMuStConvexCondit}
F_k(\boldsymbol{W'}) - F_k({\boldsymbol{W}}) \ge \langle \boldsymbol{W'} - {\boldsymbol{W}} , \nabla F_k ({\boldsymbol{W}}) \rangle + \frac{\alpha}{2} \left\| \boldsymbol{W'} - {\boldsymbol{W}} \right\|^2_2, \\
~~~~~~~~~~~~~\quad \forall k \in \mathcal{K}.      
\end{aligned}
\end{equation}
\end{assumption}

\begin{assumption}\label{AssumpBoundedVarGradient}
The expected squared $l_2$-norm of the local gradients is bounded as \cite{amiri2021convergence}:
\begin{align}\label{ConvNorm2Bound}
\mathbb{E}_{\mathfrak{D}} \left [ \left\| \nabla F_k \left( \boldsymbol{{W}}_k (r), \mathfrak{D}_k (r) \right) \right\|^2_2 \right] \le \varrho^2, \quad \forall k \in \mathcal{K}, \; \forall r.      
\end{align}
\end{assumption}
\noindent where $\mathfrak{D}_k (r)$ is a mini-batch sample at each local iteration. {It is worth mentioning that all those assumptions, when combined, allow for the analysis of the learning algorithm's convergence rate, which is an important measure of its efficiency as seen next in \Cref{convergence_detials}. They also aid in determining the conditions under which the proposed algorithm is guaranteed to reach a solution.}

\subsection{Convergence Details}
\label{convergence_detials}
As seen in Section \ref{Sec:proposed}, our proposed approach has two scheduling phases. First, all active and available clients have equal chances to take part in the global model training task. Second, for a particular cluster that reaches the stopping point, only one unique participant providing the lowest latency is selected to perform further updates for the FEEL model. Hence, for the first phase, the probability of selecting any client in every r-\textit{th} round is $\frac{{|\Omega^m_r|}}{K} = 1$. 
Now, we find the relationship between the participating clients, the number of epochs and data samples, and the convergence rate. We follow the steps similar to \cite{amiri2021convergence}. However, in our work, we consider the CFL as model training and the mini-batch SGD as a local solver where the number of updates performed locally is proportional to the number of data samples. It is worth noting that we consider  unbalanced data distribution where the number of examples (i.e., data points) is randomly distributed following the power law. Consequently, the number of local iterations can be calculated as in \eqref{local_updates}.

\begin{theorem}\label{th:1}
Given a learning rate $\eta_r = \frac{1}{\alpha \mathcal{T}} $, $\forall r$, we have:
\begin{subequations}\label{Convth:1}
\begin{equation}
\begin{aligned}\label{Convth:1_main}
\mathbb{E} \left[ \left\| \boldsymbol{{W}} (r) - {\boldsymbol{{W}}}^* \right\|_2^2 \right] \le  \Bigg( \prod_{t=0}^{r-1} \zeta_1(t) \Bigg) \left\| {\boldsymbol{{W}}}_0 - {\boldsymbol{{W}}}^* \right\|_2^2  \\
~~~~~~~~~~~~~ + \sum_{t'=0}^{r-1} \zeta_2(t') \prod_{t=t'+1}^{r-1} \zeta_1(t),  
\end{aligned}
\end{equation}
where 
\begin{align}\label{Convth:1_AB}
\zeta_1(t) \triangleq & 1 - \alpha \eta_t \left( \mathcal{T} - \eta_t (\mathcal{T} - 1) \right),\\ 
\zeta_2(t) \triangleq & \left( 1+ \alpha (1- \eta \zeta_1(t)) \right) \eta^2(t) \nonumber \\
& \varrho^2 \frac{\mathcal{T} (\mathcal{T}-1)(2\mathcal{T}-1)}{6} + \eta^2(t) (\mathcal{T}^2 + \mathcal{T}-1) \varrho^2  \nonumber \\ & + 2 \eta \zeta_1(t) (\mathcal{T} - 1) \mathfrak{F}. 
\end{align}

\end{subequations}
\end{theorem}
\begin{proof}
See Appendix \ref{AppProofth:1}. 
\end{proof}
\begin{corollary}\label{CorrConvF_FstarKequalM}
From the $\beta$-smoothness of the loss function $F(\cdot)$, after $R$ global rounds, for each cluster's model,  we have:
\begin{align}\label{ConvF_FstarKequalM}
\mathbb{E} \left[ F( \boldsymbol{{W}} (R)) \right] - F^* \le & \frac{\beta}{2} \mathbb{E} \left[ \left\| \boldsymbol{{W}} (R) - {\boldsymbol{{W}}}^* \right\|_2^2 \right] \nonumber \\
\le & \frac{\beta}{2} \prod_{t=0}^{R-1} \zeta_1(t) \left\| {\boldsymbol{{W}}} _0 - {\boldsymbol{{W}}}^* \right\|_2^2  \nonumber \\ & + \frac{\beta}{2} \sum_{t'=0}^{R-1} \zeta_2(t') \prod_{t=t'+1}^{R-1} \zeta_1(t). 
\end{align}
\end{corollary}
We can notice that the last inequality results from Theorem \ref{th:1} \cite{amiri2021convergence}. If the learning rate is continuously decreasing, $\mathop {\lim }\limits_{r \to \infty } \eta_r = 0$, and we can simply infer that $\mathop {\lim }\limits_{R \to \infty } \mathbb{E} \left[ F( \boldsymbol{{W}} (R)) \right] - F^* = 0$.
\begin{remark}
 For the client scheduling in our proposed approach, the client scheduling is not random, which brings out the advantage of capturing the data distribution amongst clients and avoiding unbiased models.
\end{remark}
\begin{remark}
For the dominant cluster that has more similarities, the average loss within its group depends on many other parameters such as $\beta$, $\alpha$, $K$, ${|\Omega^m_r|}$, $\varepsilon_1$ and $\varepsilon_2$.
\end{remark}

\section{Numerical Experiments}
\label{sec:experiment}
In this section, we perform extensive simulations to evaluate the proposed algorithms. The system models described in Section \ref{sysmodel} are properly considered with the following experiment details.

\subsection{Experimental Setup}
In all experiments, we use a bandwidth of $B=10$ MHz, with each sub-channel having a bandwidth of $1$ MHz. The channel gain of each device, $h^k_r$, is randomly modeled with a path loss ($\alpha = g_0(\frac{d_0}{d})^4$) , where $g_0 = -35 $ dB and the baseline distance $d_0 = 2$ m. We assume that the distances between the devices and the coordinating server are uniformly distributed between $20$ and $100$ m. In addition, the power of AWGN is defined as $N_0=10^{-6}$watts. The transmission power $P^{k}_r$ is distributed at random between  $p_{min} = -10$ dBm and $p_{max} = 20$ dBm. The device's CPU frequency $f_k$ is generated randomly between $1$ GHz and $9$ GHz, and the required CPU cycles per data point, $\phi$, is set to $20$ and it is assumed to be homogeneous for all devices. We use two federated datasets, FEMNIST and CIFAR-10~\cite{caldas2018leaf}, for handwriting classification and object recognition, respectively. Specifically, FEMNIST is utilized for handwriting classifications of both letters and digits (A-Z, a-z, and 0-9), and it has 305,654, 28 x 28, images while CIFAR-10 consists of 60,000 32x32 colored images. We split both datasets for each device into 80\% for training and 2\%  for testing. {Further, we utilize both datasets in a non-i.i.d. manner, whereas we split each dataset into $\mathcal{I}$ fragments, and then we assign each device only two random classes.} For the model architecture, the convolutional neural network (CNN) classifier is adopted for FEMNIST, and Alexnet—a deep neural network (DNN) is adopted for CIFAR-10. {In our models, we employed 2 hidden layers in the FEMNIST model and 13 hidden layers in the CIFAR-10 model. For both learning tasks, we utilized the ReLU activation function for the hidden layers and the Softmax activation function for the output layer.} To simplify the presentation, we list all simulation parameters in Table \ref{tab:Sim_param}.
\begin{table}[htbp]
 \footnotesize
 \centering
 \caption{{Simulation parameters}}
 \label{tab:Sim_param}
 \begin{tabular}{c|c|p{2.7cm}}
 \hline
 \thead{\textbf{\textit{Sym.}}} & \thead{\textbf{\textit{Parameter}}} & \thead{\textbf{\textit{Value(s)}}}
 \\ \hline
 $K$ & \# of participating devices & 
 (20, 50, 100, 200)
 \\ \hline
 $R$ & \# of training rounds & 200 for FEMNIST and 500 for CIFAR-10 
 \\ \hline
 $E$ & \# of local training epochs & 10
 \\ \hline
 $b$ & Batch size & 32
 \\ \hline
 $\eta$ & Learning rate & 0.01, 0.009\\
 \hline
 $B$ & Bandwidth & 10 MHz \\
 \hline
 $P_{min}$ & Minimum transmission power & -10 dBm \\ \hline
  $P_{max}$ & Maximum transmission power & 20 dBm \\ \hline
   $N_0$ & Background noise & -10 dBm \\ \hline
    $P_{k}^{min}$ & Minimum CPU frequency & 1 GHz \\ \hline
  $P_{k}^{max}$ & Maximum CPU frequency & 9 GHz \\ \hline
 \end{tabular}
\end{table}

\subsection{Benchmarks}
In this paper, we evaluate our proposed algorithms against the following state-of-the-art algorithms: 
\begin{itemize}
\item \textbf{Random approach} \cite{FedGroupDuan,sattler2020clustered}: In this approach,
the clients are scheduled randomly in each round regardless of the number of local samples, the computation and communication latencies, or the quality of the update.
 \item \textbf{Best channel scheduling}\cite{amiri2021convergence}: This approach aims to select the clients with the best channel states without taking into account the impacts of the local updates on the global model.
\item \textbf{Best local updates}\cite{amiri2021convergence}: This approach seeks to determine the participating devices with the maximum L2-norm where the participants first calculate their gradients, then find the Euclidean distance between the local parameters and the parameters of the last received model from the server. After that, they notify the server with their status, which in turn keeps the connection with the participants' having maximum L2-norm.
 \item \textbf{Maximum number of data points}: In this algorithm, the coordinating server chooses the participants having the largest data sizes and keeps connecting with them until the end of the training. We also extend this approach to select the maximum number of data examples amongst clients every round, assuming that the channel of each participant is not stable and diverse participants may be involved in different rounds. 
\end{itemize}
\subsection{Numerical Results}
To evaluate the feasibility and performance of our proposed approach, we present and compare the findings against the benchmark algorithms. To ensure a fair comparison, we utilize similar experimental setups and the hyper-parameters for all scheduling approaches regarding the number of local epochs, model structure, learning rate, and data distribution (i.e., we employed non-i.i.d. and unbalanced data distribution for both). For all experiments, we adjust $R$, the global training rounds, to 200 and 500 for FEMNIST and CIFAR-10, respectively. We also set $E = 10$, and the batch-size $b= 32$. First, we carried out the experiments and reported the results during each global round. Then, we test the{resulting} models after finishing the training process to showcase the effectiveness of the proposed approach to equipping each device with the best-fitting model. It is important to emphasize that all findings are averaged over five {trials}. 
\begin{figure*}[htbp]
\centering
         \begin{subfigure}[b]{0.6\textwidth}
         \centering
         \includegraphics[width=\textwidth]{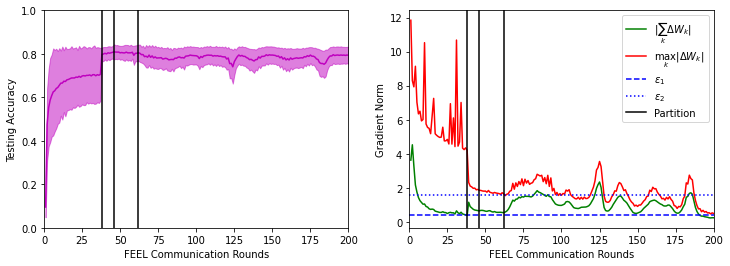}
         \caption{Proposed.}
         \label{F:a}
     \end{subfigure}
     \hfill
         \begin{subfigure}[b]{0.6\textwidth}
         \centering
         \includegraphics[width=\textwidth]{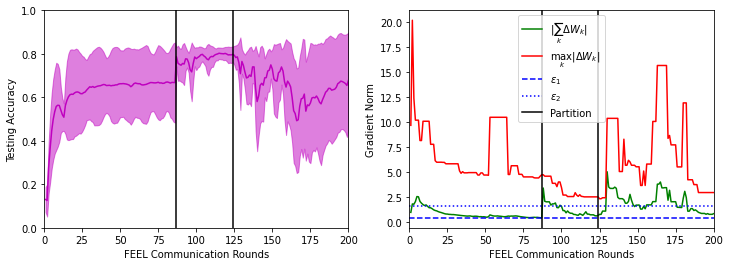}
         \caption{Baseline 1.}
         \label{F:b}
     \end{subfigure}
     \hfill     
    \caption{Average testing accuracy of the clusters'
models during the global training rounds for both our proposed approach, (a: left), and the benchmarks (b: left). In (a: right and b: right) the gradient norm of global and local loss functions during the global training rounds for both the proposed algorithms and benchmarks, respectively.}
\label{F:acc}
\end{figure*}
\subsubsection{Performance Gain on Clustering Speed and Convergence Rate}
To evaluate the proposed algorithms on the splitting and correct clustering, we start conducting the experiments using FEMNIST. We use the random selection as a baseline in this stage as it shows the best performance among benchmarks.  
\Cref{F:acc} exhibits the performance of the proposed algorithm and the random selection algorithm in terms of convergence rate and clustering speed. We use the bordered accuracy with the confidence threshold at each round (\Cref{F:a,F:b}, left) and the gradients' norm among all participants to showcase the acceleration of the convergence rate (\Cref{F:a,F:b}, right). From both figures, it is clear to note that our proposed approach (fairness followed by greedy approach) gains a much faster convergence rate regarding the speed of clustering since the partitioning starts at the training round of index 37, whereas the partitioning using the random scheduling algorithm starts at the training round of index 83. This confirms that the acceleration rate of the proposed algorithm is 2x faster than the benchmarks.  It is noted that in our approach, the partitioning actually takes place again for all participants with incongruent distributions at global training rounds of index 45 and 63, respectively until no further clustering is possible. 
All trained models achieve the stopping criteria as in (21) at round 190, proving that clustering is correctly performed, according to \Cref{F:a} right.
On the other hand, the baseline algorithm, \Cref{F:b} right, shows that the CFL still requires many more rounds to reach the stopping point at round 200.
This is owing to the random client scheduling behavior, which may choose clients who have previously been grouped, bringing out the need for many more rounds to capture the data amongst clients to perform the correct clustering. For example, the clustering becomes slow if clients with different data distributions are randomly selected at the latter rounds. 

\begin{table*}[htbp]
\caption{Models' testing accuracy after finishing all global training rounds: a conventional FL model and the tailored models of all groups (the proposed and benchmark algorithms), P*=Participant, M*=Model. (FEMNIST)}
\label{basetable}
 \begin{subtable}{\linewidth}
 \footnotesize
 \centering
 \caption{Proposed Approach.}
   \label{table_proposedsetmodel}
\begin{tabular}{|l|l|l|l|l|l|l|l|l|l|l|l|l|l|l|l|l|l|l|l|l}
 \hline
  & \textbf{P 1} & \textbf{P 2} & \textbf{P 3} & \textbf{P 4} & \textbf{P 5} & \textbf{P 6} & \textbf{P 7} & \textbf{P 8} & \textbf{P 9} & \textbf{P 10} & \textbf{P 11} & \textbf{P 12} & \textbf{P 13} & \textbf{P 14} & \textbf{P 15} \\ \hline
{\textbf{\textit{Conventional FL M}}} & 45.8    & 38   & 46   & 45   & 36.9   & 76    & 81   & 82.9   & 78.7    & 78.9    & 76.95   & 77    & 78   & 77    & 76         \\ \hline
{\textbf{\textit{M 1}}}  & 0        & 0        & 0        & 0        & 0        & \cellcolor{green!40}\textbf{77.1}   & 83.3   & \cellcolor{green!40}\textbf{86}   & \cellcolor{green!40}\textbf{81.6}   & 82    & \cellcolor{green!40}\textbf{81.8}    & \cellcolor{green!40}\textbf{81.5}    & 82.5    & 79.6     & 80.9    \\ \hline
{\textbf{\textit{M 2}}}  & 0        & 0        & \cellcolor{green!40}\textbf{81.6}   & \cellcolor{green!40}\textbf{76}   & 77.6   & 0        & \cellcolor{green!40}\textbf{83.8}   & 85.5   & 0        & \cellcolor{green!40}\textbf{82.7}    & 80.1    & 0         & \cellcolor{green!40}\textbf{84.5}   & \cellcolor{green!40}\textbf{82}   & \cellcolor{green!40}\textbf{81.8}           \\ \hline
{\textbf{\textit{M 3}}}  & \cellcolor{green!40}\textbf{83.3}   & 0        & 77   & 67   & \cellcolor{green!40}\textbf{77.9}   & 76.8   & 0        & 0        & 0        & \cellcolor{green!40}\textbf{82.7}    & 0         & 0         & 0         & 0         & 0                 \\ \hline
{\textbf{\textit{M 4}}}  & 0        & 74.5   & 0        & 0        & 0        & 0        & 82.5   & 83   & 75.2   & 0         & 0         & 0         & 0         & 0         & 0                 \\ \hline
{\textbf{\textit{M 5}}}  & \cellcolor{green!40}\textbf{83.3}  & \cellcolor{green!40}\textbf{78.6}   & 0        & 0        & 0        & 76.8   & 0        & 0        & 75.2   & 0         & 0         & 76.8    & 0         & 0         & 0               \\ \hline
{\textbf{\textit{M 6}}}  & \cellcolor{green!40}\textbf{83.3}   & 74.5   & 76.9   & 67.2   & \cellcolor{green!40}\textbf{77.9}   & 0        & 0        & 0        & 0        & 0         & 0         & 0         & 0         & 0         & 0          \\ \hline 
\textit{\textbf{Max Acc}}  & \textbf{83.3}              & \textbf{78.6}              & \textbf{81.6}              & \textbf{76}              & \textbf{77.9}              & \textbf{77.1}                         & \textbf{83.8}                          & \textbf{86}                          & \textbf{81.6}                 &\textbf{82.7}           &\textbf{81.8}  & \textbf{81.5}                            & \textbf{84.5}                             & \textbf{82}                             & \textbf{81.8}                    \\ \hline 
\end{tabular}
\end{subtable} 

 \begin{subtable}{\linewidth}
 \footnotesize
 \centering
  \caption{Benchmark 1 (Random Scheduling Approach).}
  \label{table_baselinesetmodel}
\begin{tabular}{|l|l|l|l|l|l|l|l|l|l|l|l|l|l|l|l|l|l|l|l|l}
\hline
                  & {\textbf{P 1}} & {\textbf{P 2}} & {\textbf{P 3}} & {\textbf{P 4}} & {\textbf{P 5}} & {\textbf{P 6}} & {\textbf{P 7}} & {\textbf{P 8}} & {\textbf{P 9}} & {\textbf{P 10}} & {\textbf{P 11}} & {\textbf{P 12}} & {\textbf{P 13}} & {\textbf{P 14}} & {\textbf{P 15}} \\ \hline 
\textit{\textbf{Conventional FL M}} & 50.7                       & 41.6                       & 49.2                       & 46.4                       & 41.5                       & 73.1                       & 80.1                       & 82.1                       & 76.5                       & 77.8                        & 77                          & 77.3                        & 76.7                        & 76.1                        & 77.7                         \\ \hline 
\textit{\textbf{M 1}}  & 0                          & 0                          & 0                          & 0                          & 0                          & 78.5                       & \cellcolor{green!40}\textbf{84.5}              & \cellcolor{green!40}\textbf{86.5}              & \cellcolor{green!40}\textbf{82.1}              & \cellcolor{green!40}\textbf{81.4}               & 81.5                        & \cellcolor{green!40}\textbf{79.9}               & \cellcolor{green!40}\textbf{81}                 & \cellcolor{green!40}\textbf{79.9}               & \cellcolor{green!40}\textbf{81.1}              \\ \hline 
\textit{\textbf{M 2}}  & 0                          & 0                          & 0                          & \cellcolor{red!40}\textbf{52.4}              & \cellcolor{red!40}\textbf{56.5}              & 0                          & 78.3                       & 75.8                       & 77.9                       & 72.7                        & 0                           & 0                           & 78.7                        & 78.2                        & 67.4                        \\ \hline 
\textit{\textbf{M 3}}  & \cellcolor{red!40}\textbf{59.7}              & \cellcolor{red!40}\textbf{61.1}              & \cellcolor{red!40}\textbf{62.5}              & 0                          & 0                          & \cellcolor{green!40}\textbf{82}                & 0                          & 0                          & 0                          & 0                           & \cellcolor{green!40}\textbf{82.8}               & 67.9                        & 0                           & 0                           & 0                           \\ \hline 
\textit{\textbf{M 4}}  & \cellcolor{red!40}\textbf{59.7}              & \cellcolor{red!40}\textbf{61.1}              & \cellcolor{red!40}\textbf{62.5}              & \cellcolor{red!40}\textbf{52.4}              & \cellcolor{red!40}\textbf{56.5}              & 0                          & 0                          & 0                          & 0                          & 0                           & 0                           & 0                           & 0                           & 0                           & 0                          \\ \hline 
\textit{\textbf{Max Acc}}  & \textbf{59.7}              & \textbf{61.1}              & \textbf{62.5}              & \textbf{52.4}              & \textbf{56.5}              & \textbf{82}                         & \textbf{84.5}                          & \textbf{86.5}                          & \textbf{82.1}                 &\textbf{81.4}           &\textbf{82.8}  & \textbf{79.9}                            & \textbf{81}                             & \textbf{79.9}                             & \textbf{81.1}                    \\ \hline 
\end{tabular}
  \end{subtable} 
\end{table*}
\subsubsection{Performance Gain in Terms of Testing Accuracy and Performance Gap}
Now, we present the results of the testing accuracy, focusing on the fairness between clients. 
We note that the training procedure for both all the algorithms, the proposed and benchmarks, has produced multiple models based on clustering, including a conventional FEEL model and a more tailored model for every group, as depicted in Table \ref{basetable}. We test all models amongst $15$ clients after finishing all the training rounds to show the effectiveness of each model depending on the resulting accuracy. One can see that our proposed approach yields three better-tailored models. As shown in Table \ref{table_proposedsetmodel}, it is clear that all participating devices attain well-fitting accuracy (outlined in green color), and inconsistency in performance (maximum accuracy and minimum accuracy in the last row of Table \ref{table_proposedsetmodel}) across all
is approximately $10\%$ only when our approach is used. In comparison, as depicted in Table \ref{table_baselinesetmodel}, almost more than $\frac{1}{3}$ of the tested participants(i.e., P 1, P 2, P 3, P 4, and P 5) achieve undesired accuracy (outlined in red color). The inconsistency in performance can reach up to $31.4\%$, demonstrating that some resulted models are still biased to the repeatedly selected devices which dominate the clustering.

\begin{figure*}[htbp]
\centering
\begin{subfigure}[b]{0.3\textwidth}
\includegraphics[width=\textwidth]{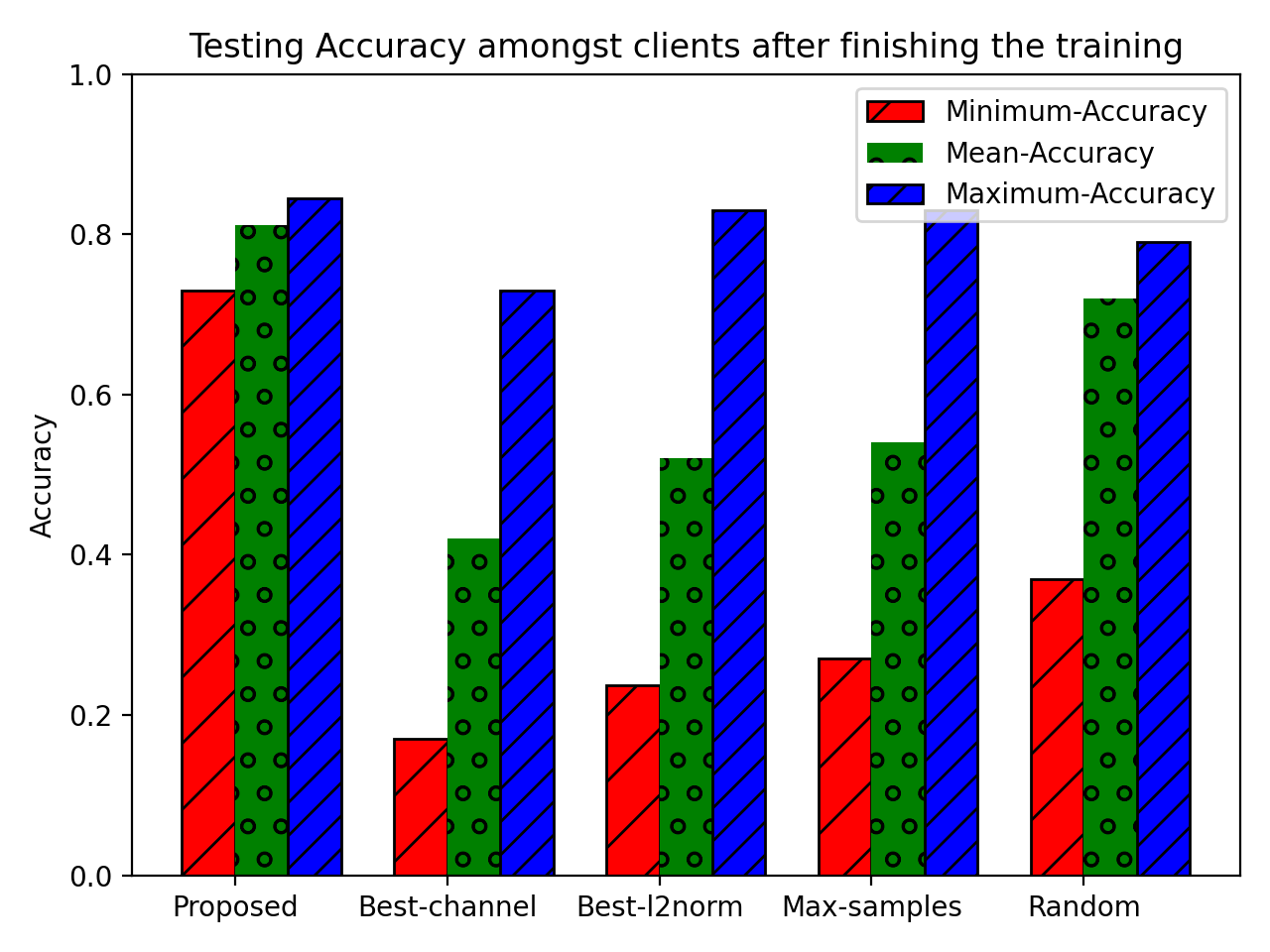}
\caption{$K$ = 20}
\label{fig:cifar_mam_20}
\end{subfigure}
\begin{subfigure}[b]{0.3\textwidth}
\includegraphics[width=\textwidth]{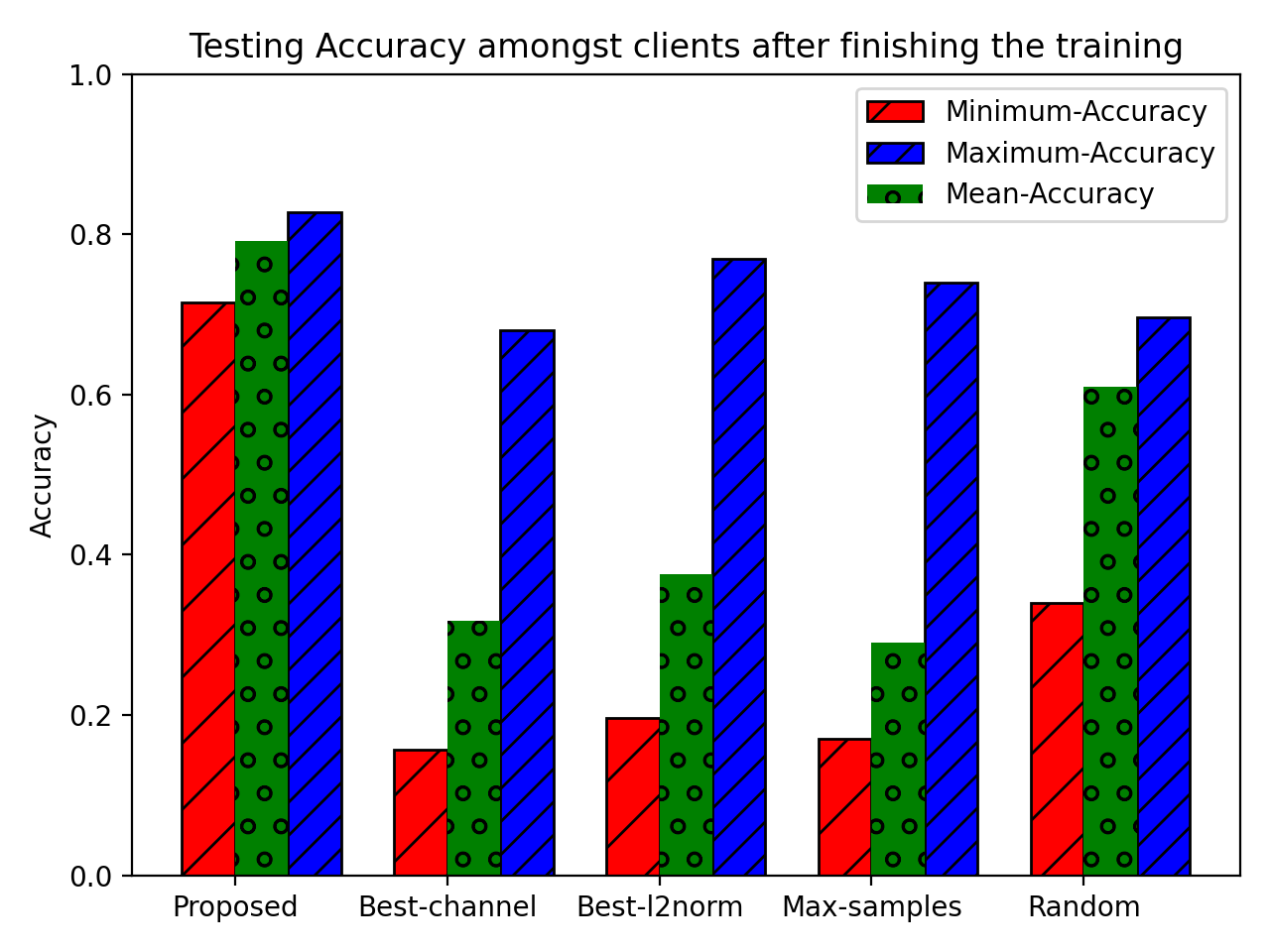}
\caption{$K$ = 50}
\label{fig:cifar_mam_50}
\end{subfigure}
\begin{subfigure}[b]{0.3\textwidth}
\includegraphics[width=\textwidth]{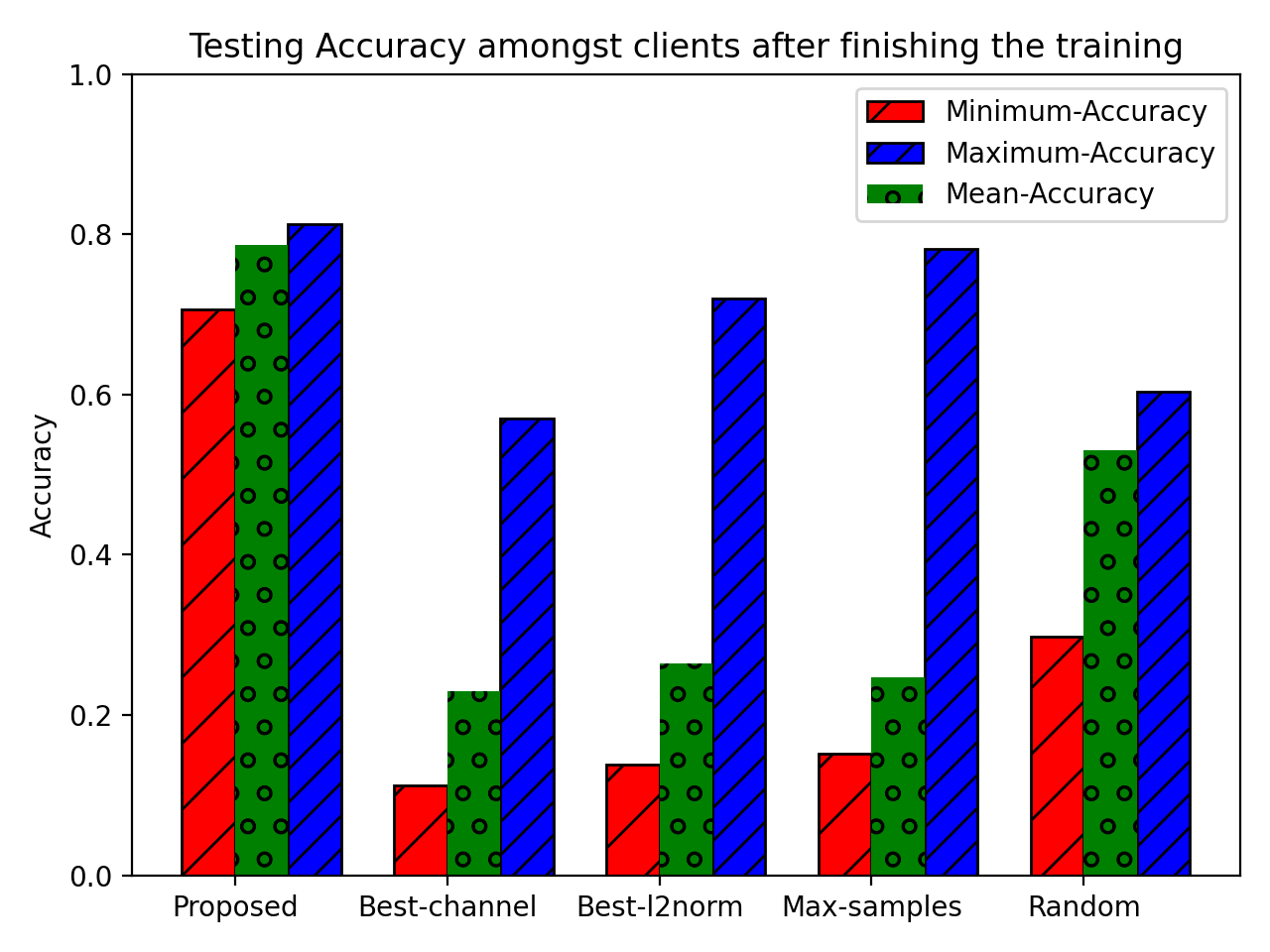}
\caption{$K$ = 100}
\label{fig:cifar_mam_100}
\end{subfigure}
\begin{subfigure}[b]{0.3\textwidth}
\includegraphics[width=\textwidth]{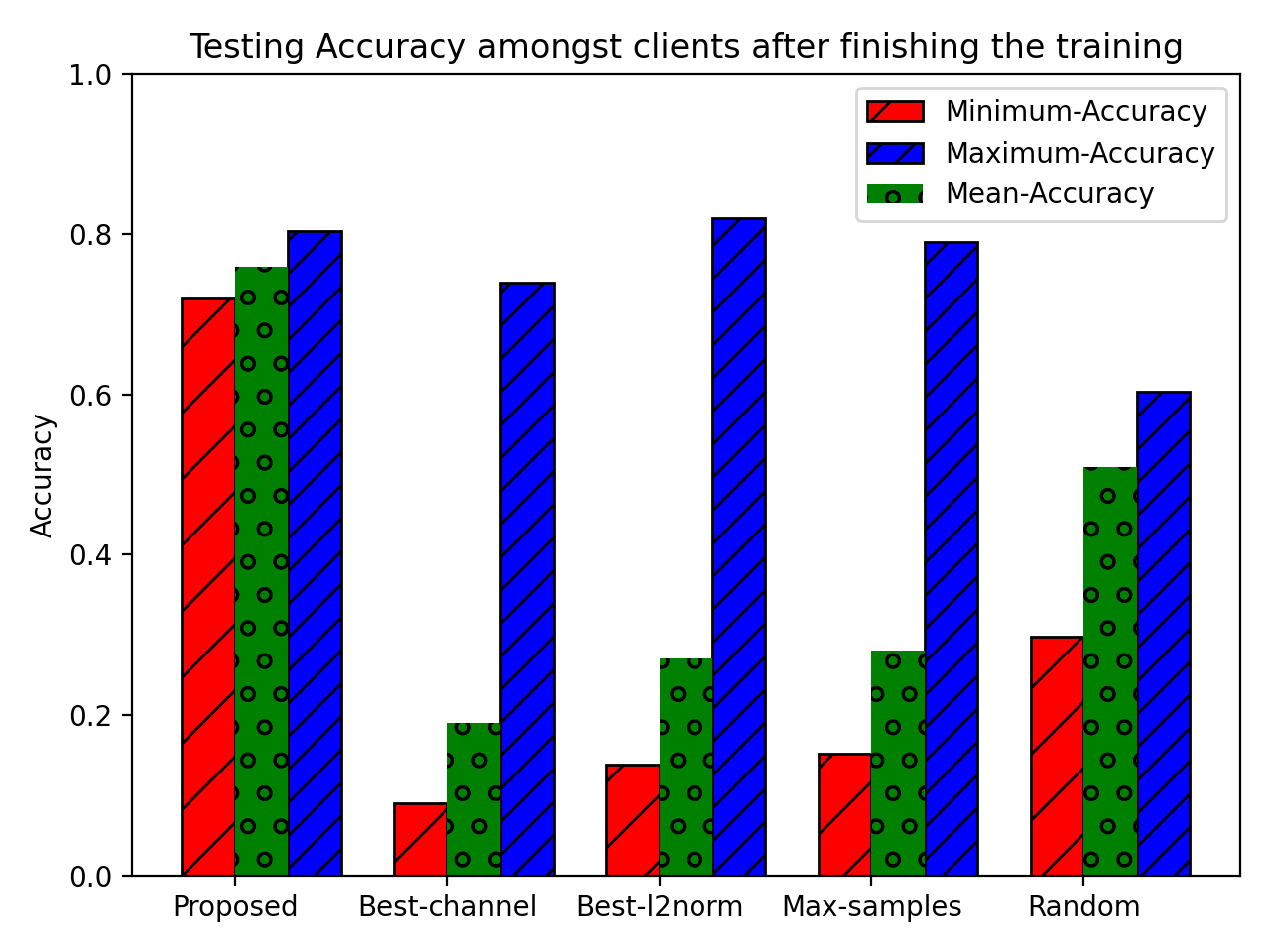}
\caption{$K$ = 200}
\label{fig:cifar_mam_200}
\end{subfigure}
\caption{Minimum, average, and maximum accuracy overall testing clients (CIFAR-10) when the data is distributed over 20, 50, 100, and 200 clients.}
\label{fig:max_avg_mean_cifar-10}
\end{figure*}

Furthermore, to validate the proposed algorithms, we perform other simulation experiments using the CIFAR-10 dataset as a complex learning task under non-i.i.d. federated settings. This assists in generalizing the performance of the proposed algorithms for different application domains. \Cref{fig:cifar_mam_20,fig:cifar_mam_50,fig:cifar_mam_100,fig:cifar_mam_200} show the minimum, maximum, and average testing accuracy amongst clients where the minimum accuracy represents the lowest accuracy can be achieved by the client and vise versa. It is worth mentioning that the proposed approach balances the accuracy between all clients despite the number of participating clients where the gap between the maximum and minimum accuracy amongst clients reaches up to $10\%$, meaning that the resulting models optimally fit all local data distribution. At the same time, we notice that the best channel, the best l2-norm, and the max-samples approaches increase the accuracy gap between the clients, reaching up to $70\%$, especially if the number of clients increases. This stems from the fact that the server stays connected with the same clients during all global rounds in those approaches, particularly when the channel is not changing rapidly, leading to biased models that can only fit the participants' clients. In addition, random scheduling provides less accuracy gaps due to the randomness of participant selection. Nevertheless, the proposed scheduling approach substantially outperforms all benchmarks and provides much more stable accuracy regardless of the distribution of the data and the number of clients in the network. 

\subsection{Lessons Learned}
The key takeaways from the experiments outlined in this paper can be listed as follows.
\begin{itemize}
    \item The proposed scheduling algorithms can accelerate the convergence rate compared to the baseline scheduling algorithms due to the fairness between clients at the beginning of the training.
    \item The proposed approach is effectively suited to deal with CFL when compared with the scheduling algorithms of traditional FL, seeking to avoid biased models to frequently selected clients.
    \item Selecting all clients at the beginning of the training leads to performing the correct clustering and capturing all data distributions amongst clients while increasing the clustering speed and accelerating the convergence rate.
    \item The proposed algorithms also reduce the resource consumption as only a single client can take part in the model training once a specific cluster converges to the stopping point.
    \item Overall, the proposed approach surpasses the benchmarks in terms of agglomeration quality with more tailored models, expediting convergence and minimizing training costs. 
\end{itemize}
\section{Conclusion}
\label{conclusion}
In this paper, we have proposed novel client scheduling and selection algorithms for clustered federated multitask learning to reduce the training latency and speed up the convergence rate, which improves the resulting model for each cluster. The proposed algorithms are based on the fairness between the clients across the network, where all available clients have equal chances of being selected to take part in the training process regardless of the channel state or the size of their local data.  This enables the edge system to imbue the clients with more specialized models rather than having biased models. Given non-i.i.d. and unbalanced data distribution, clients' heterogeneity, and restricted bandwidth, we have first formulated an optimization problem to obtain the best client scheduling that minimizes the training cost and improves the convergence speed. We have analyzed the relationship between the proposed scheduling approach and the convergence rate of the specialized models. We have conducted extensive simulation experiments using realistic federated datasets, FEMNIST and CIFAR-10. The findings demonstrate that the proposed approach efficiently reduces the training latency and accelerates convergence while attaining a satisfying performance. For future research, it will be interesting to adapt the age of the updates for the missing clients' updates and find the optimal thresholds for the splitting conditions.

\section*{Acknowledgement}
The work of Abdullatif Albsaeer, Mohamed Abdallah, and Ala Al-fuqaha was made possible by NPRP-Standard (NPRP-S) Thirteen (13th) Cycle grant \# NPRP13S-0201-200219 from the Qatar National Research Fund (a member of Qatar Foundation) and the work of O. A. Dobre was partially supported by the natural Sciences and Engineering research Council of Canada (NSERC) through its Discovery program. The findings herein reflect the work and are solely the responsibility of the authors.
\bibliographystyle{IEEEtran}
\bibliography{ref}

\begin{thebibliography}{10}
\providecommand{\url}[1]{#1}
\csname url@samestyle\endcsname
\providecommand{\newblock}{\relax}
\providecommand{\bibinfo}[2]{#2}
\providecommand{\BIBentrySTDinterwordspacing}{\spaceskip=0pt\relax}
\providecommand{\BIBentryALTinterwordstretchfactor}{4}
\providecommand{\BIBentryALTinterwordspacing}{\spaceskip=\fontdimen2\font plus
\BIBentryALTinterwordstretchfactor\fontdimen3\font minus
  \fontdimen4\font\relax}
\providecommand{\BIBforeignlanguage}[2]{{%
\expandafter\ifx\csname l@#1\endcsname\relax
\typeout{** WARNING: IEEEtran.bst: No hyphenation pattern has been}%
\typeout{** loaded for the language `#1'. Using the pattern for}%
\typeout{** the default language instead.}%
\else
\language=\csname l@#1\endcsname
\fi
#2}}
\providecommand{\BIBdecl}{\relax}
\BIBdecl

\bibitem{albaseer2021client}
A.~Albaseer, M.~Abdallah, A.~Al-Fuqaha, and A.~Erbad, ``Client selection
  approach in support of clustered federated learning over wireless edge
  networks,'' in \emph{2021 IEEE Global Communications Conference (GLOBECOM)},
  2021, pp. 1--6.

\bibitem{chen2022joint}
Y.~Chen, Y.~Sun, B.~Yang, and T.~Taleb, ``Joint caching and computing service
  placement for edge-enabled iot based on deep reinforcement learning,''
  \emph{IEEE Internet of Things Journal}, 2022.

\bibitem{chen2022dynamic}
Y.~Chen, Y.~Sun, C.~Wang, and T.~Taleb, ``Dynamic task allocation and service
  migration in edge-cloud iot system based on deep reinforcement learning,''
  \emph{IEEE Internet of Things Journal}, 2022.

\bibitem{luo2020power}
Y.~Luo, J.~Yang, W.~Xu, K.~Wang, and M.~Di~Renzo, ``Power consumption
  optimization using gradient boosting aided deep q-network in c-rans,''
  \emph{IEEE Access}, vol.~8, pp. 46\,811--46\,823, March 2020.

\bibitem{9352033}
O.~A. Wahab, A.~Mourad, H.~Otrok, and T.~Taleb, ``Federated machine learning:
  Survey, multi-level classification, desirable criteria and future directions
  in communication and networking systems,'' \emph{IEEE Communications Surveys
  \& Tutorials}, vol.~23, no.~2, pp. 1342--1397, 2021.

\bibitem{9148111}
B.~S. Ciftler, A.~Albaseer, N.~Lasla, and M.~Abdallah, ``Federated learning for
  rss fingerprint-based localization: A privacy-preserving crowdsourcing
  method,'' in \emph{2020 International Wireless Communications and Mobile
  Computing (IWCMC)}, 2020, pp. 2112--2117.

\bibitem{9220780}
S.~Abdulrahman, H.~Tout, H.~Ould-Slimane, A.~Mourad, C.~Talhi, and M.~Guizani,
  ``A survey on federated learning: The journey from centralized to distributed
  on-site learning and beyond,'' \emph{IEEE Internet of Things Journal},
  vol.~8, no.~7, pp. 5476--5497, 2021.

\bibitem{mcmahan2016communication}
B.~McMahan, E.~Moore, D.~Ramage, S.~Hampson, and B.~A. y~Arcas,
  ``Communication-efficient learning of deep networks from decentralized
  data,'' in \emph{Artificial intelligence and statistics}.\hskip 1em plus
  0.5em minus 0.4em\relax PMLR, 2017, pp. 1273--1282.

\bibitem{9148475}
A.~Albaseer, B.~S. Ciftler, M.~Abdallah, and A.~Al-Fuqaha, ``Exploiting
  unlabeled data in smart cities using federated edge learning,'' in \emph{2020
  International Wireless Communications and Mobile Computing (IWCMC)}, 2020,
  pp. 1666--1671.

\bibitem{10060393}
A.~Albaseer and M.~Abdallah, ``Privacy-preserving honeypot-based detector in
  smart grid networks: A new design for quality-assurance and fair incentives
  federated learning framework,'' in \emph{2023 IEEE 20th Consumer
  Communications \& Networking Conference (CCNC)}, 2023, pp. 722--727.

\bibitem{wang2020federated}
X.~Wang, C.~Wang, X.~Li, V.~C. Leung, and T.~Taleb, ``Federated deep
  reinforcement learning for internet of things with decentralized cooperative
  edge caching,'' \emph{IEEE Internet of Things Journal}, vol.~7, no.~10, pp.
  9441--9455, 2020.

\bibitem{9851847}
K.~D. Stergiou and K.~E. Psannis, ``Federated learning approach decouples
  clients from training a local model and with the communication with the
  server,'' \emph{IEEE Transactions on Network and Service Management}, pp.
  1--1, 2022.

\bibitem{zhu2021data}
Z.~Zhu, J.~Hong, and J.~Zhou, ``Data-free knowledge distillation for
  heterogeneous federated learning,'' \emph{arXiv preprint arXiv:2105.10056},
  2021.

\bibitem{pang2020realizing}
J.~Pang, Y.~Huang, Z.~Xie, Q.~Han, and Z.~Cai, ``Realizing the heterogeneity: A
  self-organized federated learning framework for iot,'' \emph{IEEE Internet of
  Things Journal}, vol.~8, no.~5, pp. 3088--3098, July 2020.

\bibitem{fallah2020personalized}
A.~Fallah, A.~Mokhtari, and A.~Ozdaglar, ``Personalized federated learning: A
  meta-learning approach,'' \emph{arXiv preprint arXiv:2002.07948}, 2020.

\bibitem{sattler2020clustered}
F.~Sattler, K.-R. M{\"u}ller, and W.~Samek, ``Clustered federated learning:
  Model-agnostic distributed multitask optimization under privacy
  constraints,'' \emph{IEEE Transactions on Neural Networks and Learning
  Systems}, vol.~32, no.~8, pp. 3710--3722, August 2020.

\bibitem{albaseer2022data}
A.~Albaseer, M.~Abdallah, A.~Al-Fuqaha \emph{et~al.}, ``Data-driven participant
  selection and bandwidth allocation for heterogeneous federated edge
  learning,'' 2022.

\bibitem{9319704}
T.~Subramanya and R.~Riggio, ``Centralized and federated learning for
  predictive vnf autoscaling in multi-domain 5g networks and beyond,''
  \emph{IEEE Transactions on Network and Service Management}, vol.~18, no.~1,
  pp. 63--78, 2021.

\bibitem{10092826}
M.~H. Mahmoud, A.~Albaseer, M.~Abdallah, and N.~Al-Dhahir, ``Federated learning
  resource optimization and client selection for total energy minimization
  under outage, latency, and bandwidth constraints with partial or no csi,''
  \emph{IEEE Open Journal of the Communications Society}, vol.~4, pp. 936--953,
  2023.

\bibitem{sattler2020byzantine}
F.~Sattler, K.-R. Müller, T.~Wiegand, and W.~Samek, ``On the byzantine
  robustness of clustered federated learning,'' in \emph{ICASSP 2020 - 2020
  IEEE International Conference on Acoustics, Speech and Signal Processing
  (ICASSP)}, 2020, pp. 8861--8865.

\bibitem{smith2017federated}
V.~Smith, C.-K. Chiang, M.~Sanjabi, and A.~S. Talwalkar, ``Federated multi-task
  learning,'' in \emph{Advances in Neural Information Processing Systems},
  December 2017, pp. 4424--4434.

\bibitem{zhang2020personalized}
M.~Zhang, K.~Sapra, S.~Fidler, S.~Yeung, and J.~M. Alvarez, ``Personalized
  federated learning with first order model optimization,'' \emph{arXiv
  preprint arXiv:2012.08565}, 2020.

\bibitem{xie2020multi}
M.~Xie, G.~Long, T.~Shen, T.~Zhou, X.~Wang, and J.~Jiang, ``Multi-center
  federated learning,'' \emph{arXiv preprint arXiv:2005.01026}, 2020.

\bibitem{ghosh2020efficient}
A.~Ghosh, J.~Chung, D.~Yin, and K.~Ramchandran, ``An efficient framework for
  clustered federated learning,'' \emph{Advances in Neural Information
  Processing Systems}, vol.~33, 2020.

\bibitem{briggs2020federated}
C.~Briggs, Z.~Fan, and P.~Andras, ``Federated learning with hierarchical
  clustering of local updates to improve training on non-iid data,'' in
  \emph{Proc. of International Joint Conference on Neural Networks (IJCNN)},
  July 2020, pp. 1--9.

\bibitem{xu2020client}
J.~Xu and H.~Wang, ``Client selection and bandwidth allocation in wireless
  federated learning networks: A long-term perspective,'' \emph{arXiv preprint
  arXiv:2004.04314}, 2020.

\bibitem{shi2020joint}
W.~Shi, S.~Zhou, Z.~Niu, M.~Jiang, and L.~Geng, ``Joint device scheduling and
  resource allocation for latency constrained wireless federated learning,''
  \emph{IEEE Transactions on Wireless Communications}, vol.~20, no.~1, pp.
  453--467, September 2020.

\bibitem{yang2019scheduling}
H.~H. Yang, Z.~Liu, T.~Q. Quek, and H.~V. Poor, ``Scheduling policies for
  federated learning in wireless networks,'' \emph{IEEE Transactions on
  Communications}, vol.~68, no.~1, pp. 317--333, September 2019.

\bibitem{chen2022federated}
X.~Chen, G.~Zhu, Y.~Deng, and Y.~Fang, ``Federated learning over multihop
  wireless networks with in-network aggregation,'' \emph{IEEE Transactions on
  Wireless Communications}, vol.~21, no.~6, pp. 4622--4634, 2022.

\bibitem{9844147}
A.~Albaseer, M.~Abdallah, A.~Al-Fuqaha, A.~Erbad, and O.~A. Dobre,
  ``Semi-supervised federated learning over heterogeneous wireless iot edge
  networks: Framework and algorithms,'' \emph{IEEE Internet of Things Journal},
  vol.~9, no.~24, pp. 25\,626--25\,642, 2022.

\bibitem{amiri2020federated}
M.~M. Amiri and D.~G{\"u}nd{\"u}z, ``Federated learning over wireless fading
  channels,'' \emph{IEEE Transactions on Wireless Communications}, vol.~19,
  no.~5, pp. 3546--3557, 2020.

\bibitem{chen2020convergence}
M.~Chen, H.~V. Poor, W.~Saad, and S.~Cui, ``Convergence time minimization of
  federated learning over wireless networks,'' in \emph{Proc. of IEEE
  International Conference on Communications (ICC)}.\hskip 1em plus 0.5em minus
  0.4em\relax IEEE, 2020, pp. 1--6.

\bibitem{feng2021design}
C.~Feng, Z.~Zhao, Y.~Wang, T.~Q. Quek, and M.~Peng, ``On the design of
  federated learning in the mobile edge computing systems,'' \emph{IEEE
  Transactions on Communications}, 2021.

\bibitem{albaseer2021threshold}
A.~Albaseer, M.~Abdallah, A.~Al-Fuqaha, and A.~Erbad, ``Threshold-based data
  exclusion approach for energy-efficient federated edge learning,'' \emph{in
  Proc. of IEEE International Conference on Communications Workshops (ICC
  Workshops)}, pp. 1--6, 2021.

\bibitem{albaseer2021fine}
A.~M. Albaseer, M.~Abdallah, A.~Al-Fuqaha, and A.~Erbad, ``Fine-grained data
  selection for improved energy efficiency of federated edge learning,''
  \emph{IEEE Transactions on Network Science and Engineering}, July 2021.

\bibitem{9416805}
M.~H.~u. Rehman, A.~M. Dirir, K.~Salah, E.~Damiani, and D.~Svetinovic,
  ``Trustfed: A framework for fair and trustworthy cross-device federated
  learning in iiot,'' \emph{IEEE Transactions on Industrial Informatics},
  vol.~17, no.~12, pp. 8485--8494, 2021.

\bibitem{9766403}
A.~Hammoud, H.~Otrok, A.~Mourad, and Z.~Dziong, ``On demand fog federations for
  horizontal federated learning in iov,'' \emph{IEEE Transactions on Network
  and Service Management}, pp. 1--1, 2022.

\bibitem{amiri2021convergence}
M.~M. Amiri, D.~G{\"u}nd{\"u}z, S.~R. Kulkarni, and H.~V. Poor, ``Convergence
  of update aware device scheduling for federated learning at the wireless
  edge,'' \emph{IEEE Transactions on Wireless Communications}, vol.~20, no.~6,
  pp. 3643--3658, Jun. 2021.

\bibitem{wu2021node}
H.~Wu and P.~Wang, ``Node selection toward faster convergence for federated
  learning on non-iid data,'' \emph{arXiv preprint arXiv:2105.07066}, 2021.

\bibitem{9809926}
Z.~Qu, R.~Duan, L.~Chen, J.~Xu, Z.~Lu, and Y.~Liu, ``Context-aware online
  client selection for hierarchical federated learning,'' \emph{IEEE
  Transactions on Parallel and Distributed Systems}, pp. 1--15, 2022.

\bibitem{9824246}
A.~Albaseer, M.~Abdallah, A.~Al-Fuqaha, and A.~Erbad, ``Balanced energy
  consumption based on historical participation of resource-constrained devices
  in federated edge learning,'' in \emph{2022 International Wireless
  Communications and Mobile Computing (IWCMC)}, 2022, pp. 300--305.

\bibitem{9766408}
T.~Huang, W.~Lin, L.~Shen, K.~Li, and A.~Y. Zomaya, ``Stochastic client
  selection for federated learning with volatile clients,'' \emph{IEEE Internet
  of Things Journal}, pp. 1--1, 2022.

\bibitem{wang2019adaptive}
S.~Wang, T.~Tuor, T.~Salonidis, K.~K. Leung, C.~Makaya, T.~He, and K.~Chan,
  ``Adaptive federated learning in resource constrained edge computing
  systems,'' \emph{IEEE Journal on Selected Areas in Communications}, vol.~37,
  no.~6, pp. 1205--1221, March 2019.

\bibitem{tran2019federated}
N.~H. Tran, W.~Bao, A.~Zomaya, and C.~S. Hong, ``Federated learning over
  wireless networks: Optimization model design and analysis,'' in \emph{Proc.
  of IEEE INFOCOM}, 2019, pp. 1387--1395.

\bibitem{caldas2018leaf}
S.~Caldas, P.~Wu, T.~Li, J.~Kone{\v{c}}n{\`y}, H.~B. McMahan, V.~Smith, and
  A.~Talwalkar, ``Leaf: A benchmark for federated settings,'' \emph{arXiv
  preprint arXiv:1812.01097}, 2018.

\bibitem{FedGroupDuan}
\BIBentryALTinterwordspacing
M.~Duan, D.~Liu, X.~Ji, R.~Liu, L.~Liang, X.~Chen, and Y.~Tan, ``Fedgroup:
  Ternary cosine similarity-based clustered federated learning framework toward
  high accuracy in heterogeneous data,'' \emph{CoRR}, vol. abs/2010.06870,
  2020. [Online]. Available: \url{https://arxiv.org/abs/2010.06870}
\BIBentrySTDinterwordspacing

\end{thebibliography}

\begin{IEEEbiography}[{\includegraphics[width=1.1in,height=1.22in,clip]{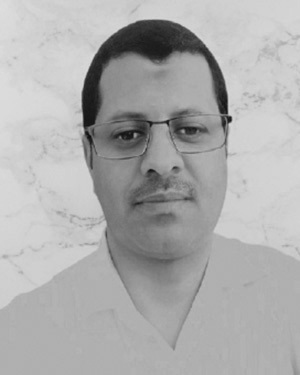}}]%
{Abdullatif Albaseer (Member, IEEE)} 
received an M.Sc. degree in computer networks from King Fahd University of Petroleum and Minerals, Dhahran, Saudi Arabia, in 2017 and a Ph.D. degree in computer science and engineering from Hamad Bin Khalifa University, Doha, Qatar, in 2022. He is a Postdoctoral Research Fellow with the Smart Cities and IoT Lab at Hamad Bin Khalifa University. He has authored and co-authored more than twenty conference and journal papers in IEEE ICC, IEEE Globecom, IEEE CCNC, and IEEE Transactions. He also has six US patents in the area of the wireless network edge. His current research interests include semantic communication, federated learning over network edge, industrial IoT, and cybersecurity. 
\end{IEEEbiography}

\begin{IEEEbiography}[{\includegraphics[width=1.1in,height=1.22in,clip]{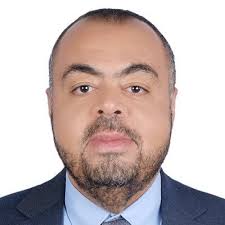}}]%
{Mohamed Abdallah (Senior Member,
IEEE)} received the B.Sc. degree from Cairo University, in 1996, and the M.Sc. and Ph.D. degrees
from the University of Maryland at College Park,
in 2001 and 2006, respectively. From 2006 to
2016, he held academic and research positions at
Cairo University and Texas A\&M University at
Qatar. He is currently a Founding Faculty Member
with the rank of Associate Professor with the
College of Science and Engineering, Hamad Bin
Khalifa University (HBKU). His current research interests include wireless networks, wireless security, smart grids, optical wireless communication,
and blockchain applications for emerging networks. He has published more than 150 journals and conferences and four book chapters, and co-invented
four patents. He was a recipient of the Research Fellow Excellence Award at Texas A\&M University at Qatar, in 2016, the Best Paper Award in multiple
IEEE conferences including the IEEE BlackSeaCom 2019, the IEEE First
Workshop on Smart Grid and Renewable Energym in 2015, and the Nortel
Networks Industrial Fellowship for five consecutive years, from 1999 to
2003. His professional activities include an Associate Editor of the IEEE
TRANSACTIONS ON COMMUNICATIONS and the IEEE OPEN ACCESS JOURNAL OF
COMMUNICATIONS, a Track Co-Chair of the IEEE VTC Fall 2019 conference,
a Technical Program Chair of the 10th International Conference on Cognitive
Radio Oriented Wireless Networks, and a Technical Program Committee
Member of several major IEEE conferences.
\end{IEEEbiography}

\begin{IEEEbiography}[{\includegraphics[width=1.1in,height=1.22in,clip]{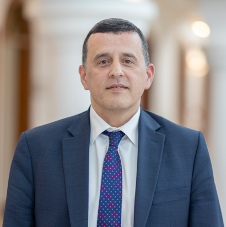}}]%
{Ala Al-Fuqaha (Senior Member, IEEE)}
  received Ph.D. degree in Computer Engineering and Networking from the University of Missouri-Kansas City, Kansas City, MO, USA, in 2004. He is currently a professor at Hamad Bin Khalifa University (HBKU). His research interests include the use of machine learning in general and deep learning in particular in support of the data-driven and self-driven management of large-scale deployments of IoT and smart city infrastructure and services, Wireless Vehicular Networks (VANETs), cooperation, and spectrum access etiquette in cognitive radio networks, and management and planning of software defined networks (SDN). He is a senior member of the IEEE and an ABET Program Evaluator (PEV). He serves on editorial boards of multiple journals including IEEE Communications Letter and IEEE Network Magazine. He also served as chair, co-chair, and technical program committee member of multiple international conferences including IEEE VTC, IEEE Globecom, IEEE ICC, and IWCMC.
\end{IEEEbiography}

\begin{IEEEbiography} [{\includegraphics[width=1in,height=1.25in,clip,keepaspectratio]{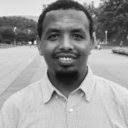}}]{ \rmfamily  Abegaz Mohammed Seid (memeber, IEEE)} \rmfamily \mdseries received his B.Sc. and M.Sc. degrees in Computer Science from Ambo University and Addis Ababa University, Ethiopia, in 2010 and 2015, respectively. He received a Ph.D. degree in Computer Science and Technology from the University of Electronic Science and Technology of China (UESTC) in 2021. He is currently a post-doctoral fellow with the  College of  Science and  Engineering at Hamad  Bin Khalifa University, Doha, Qatar. He served as a graduate assistant and lecturer, as well as a member of the academic committee and an associate registrar at Dilla University, Ethiopia, from 2010 to 2016. Dr. Abegaz has published more than thirteen scientific conferences and journal papers. His research interests include a wireless network, mobile edge computing, blockchain, machine learning, Vehicular network, IoT, machine learning, UAV Network, IoT, and 5G/6G wireless network.
\end{IEEEbiography}

\begin{IEEEbiography}
[{\includegraphics[width=1in,height=1.25in,clip,keepaspectratio]{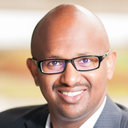}}]{ \rmfamily  Aiman Erbad (Senior Member, IEEE)} \rmfamily \mdseries
 is an Associate Professor and ICT Division Head in the College of Science and Engineering, Hamad Bin Khalifa University (HBKU). Prior to this, he was an Associate Professor at the Computer Science and Engineering (CSE) Department and the Director of Research Planning and Development at Qatar University until May 2020. He also served as the Director of Research Support responsible for all grants and contracts (2016–2018) and as the Computer Engineering Program Coordinator (2014–2016). Dr. Erbad obtained a Ph.D. in Computer Science from the University of British Columbia (Canada) in 2012, a Master of Computer Science in embedded systems and robotics from the University of Essex (UK) in 2005, and a B.Sc. in Computer Engineering from the University of Washington, Seattle in 2004. He received the Platinum award from H.H. The Emir Sheikh Tamim bin Hamad Al Thani at the Education Excellence Day 2013 (Ph.D. category). He also received the 2020 Best Research Paper Award from Computer Communications, the IWCMC 2019 Best Paper Award, and the IEEE CCWC 2017 Best Paper Award. His research received funding from the Qatar National Research Fund, and his research outcomes were published in respected international conferences and journals. He is an editor for KSII Transactions on Internet and Information Systems, an editor for the International Journal of Sensor Networks (IJSNet), and a guest editor for IEEE Network. He also served as a Program Chair of the International Wireless Communications Mobile Computing Conference (IWCMC 2019), as a Publicity chair of the ACM MoVid Workshop 2015, as a Local Arrangement Chair of NOSSDAV 2011, and as a Technical Program Committee (TPC) member in various IEEE and ACM international conferences (GlobeCom, NOSSDAV, MMSys, ACMMM, IC2E, and ICNC). His research interests span cloud computing, edge intelligence, Internet of Things (IoT), private and secure networks, and multimedia systems. He is a senior member of IEEE and ACM.
\end{IEEEbiography}

\begin{IEEEbiography}[{\includegraphics[width=1.1in,height=1.22in,clip]{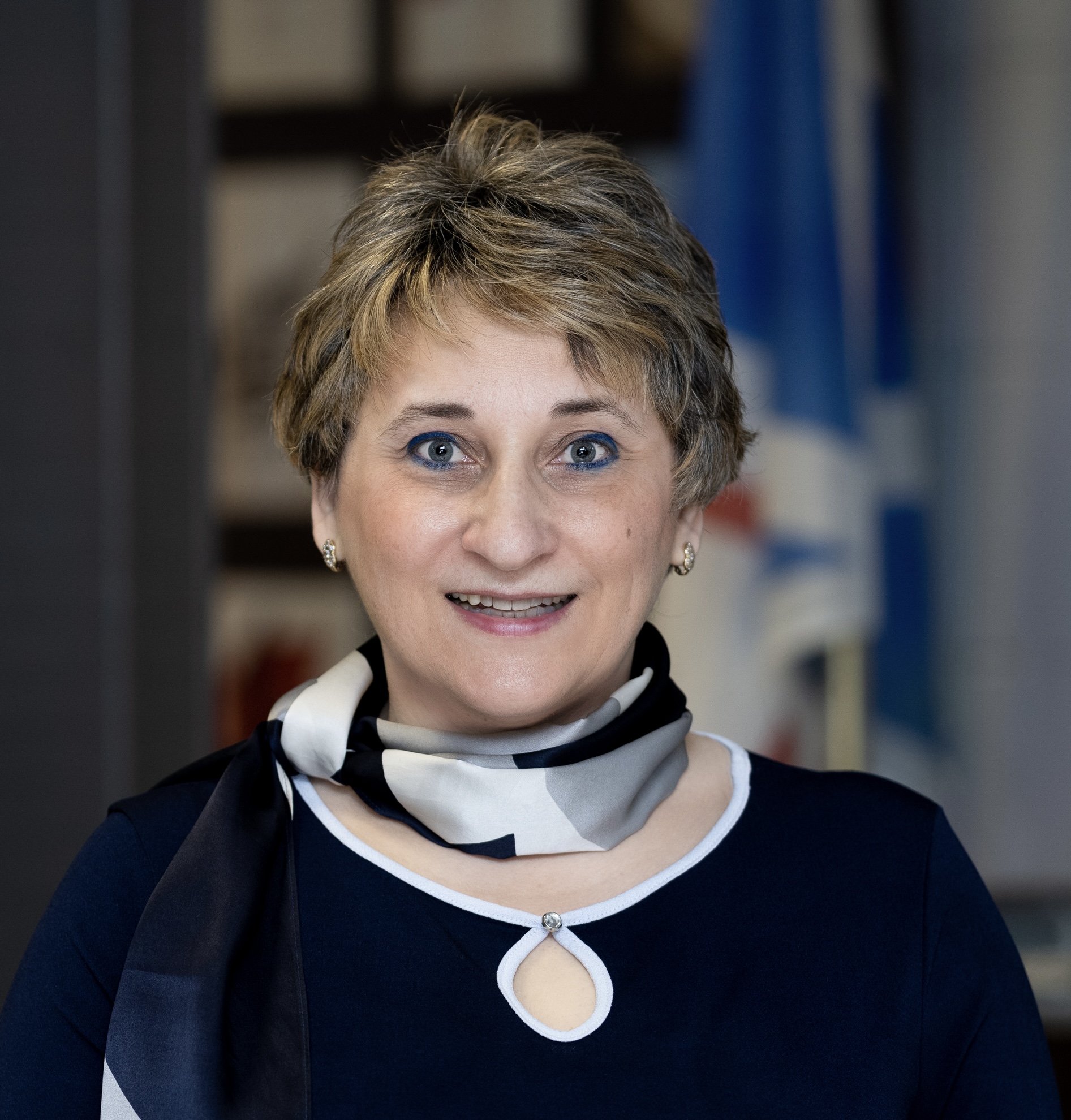}}]%
{Octavia A. Dobre (Fellow, IEEE)} received
the Dipl. Ing. and Ph.D. degrees from the Polytechnic Institute of Bucharest, Romania, in 1991
and 2000, respectively. Between 2002 and 2005,
she was with New Jersey Institute of Technology,
USA. In 2005, she joined Memorial University,
Canada, where she is currently a Professor and
Research Chair. She was a Visiting Professor with
Massachusetts Institute of Technology, USA and
Universite de Bretagne Occidentale, France. Her ´
research interests encompass wireless, optical and
underwater communication technologies. She has (co-)authored over 400
refereed papers in these areas.
Dr. Dobre serves as the Editor-in-Chief (EiC) of the IEEE Open Journal of
the Communications Society. She was the EiC of the IEEE Communications
Letters, Senior Editor, Editor, and Guest Editor for various prestigious journals
and magazines. She also served as General Chair, Technical Program CoChair, Tutorial Co-Chair, and Technical Co-Chair of symposia at numerous
conferences.
Dr. Dobre was a Fulbright Scholar, Royal Society Scholar, and Distinguished Lecturer of the IEEE Communications Society. She obtained Best
Paper Awards at various conferences, including IEEE ICC, IEEE Globecom,
IEEE WCNC, and IEEE PIMRC. Dr. Dobre is a Fellow of the Engineering
Institute of Canada and a Fellow of the Canadian Academy of Engineering.
\end{IEEEbiography}
\onecolumn
\appendices
\section{Proof of Theorem \ref{th:1}}\label{AppProofth:1}
We recall the updated global model as follows:
\begin{align}\label{AppFDPModUpdateTheta}
{\boldsymbol{W}} (r+1) = {\boldsymbol{W}} (r) + \frac{1}{K} \sum\nolimits_{k \in \Omega^m_r} \widehat{\boldsymbol{\mathcal{L}}}_k (r),
\end{align}
Let us define the following aux variable as in \cite{amiri2021convergence}:
\begin{align}\label{AppFDPThetaTilde}
\boldsymbol{W'} (r+1) = {\boldsymbol{W}} (r) + \frac{1}{K} \sum\nolimits_{k =1}^{K} \widehat{\boldsymbol{\mathcal{L}}}_k (r).
\end{align}
We can have: 
\begin{align}\label{AppFDP_1}
& \left\| {\boldsymbol{W}} (r+1) - {{\boldsymbol{W}}}^* \right\|_2^2 =  \left\| {\boldsymbol{W}} (r+1) - {\boldsymbol{W'}} (r+1) + {\boldsymbol{W'}} (r+1)  - {{\boldsymbol{W}}}^* \right\|_2^2 \nonumber \\
& \; = \left\| {\boldsymbol{W}} (r+1) - {\boldsymbol{W'}} (r+1) \right\|_2^2 + \left\| {\boldsymbol{W'}} (r+1)  - {{\boldsymbol{W}}}^* \right\|_2^2  + 2 \langle {\boldsymbol{W}} (r+1) - {\boldsymbol{W'}} (r+1) , {\boldsymbol{W'}} (r+1)  - {{\boldsymbol{W}}}^* \rangle.  
\end{align}
Now, we bound the average of the right hand side of \eqref{AppFDP_1}.

\begin{lemma}\label{AppLemmaTerm_1}
We have the following optimally gap for the fairness scheduling algorithm:
\begin{align}\label{AppLemmaTerm_1_Eq_1}
\mathbb{E} \left[ \left\| {\boldsymbol{W}} (r+1) - {\boldsymbol{W'}} (r+1) \right\|_2^2 \right] \le \epsilon.    
\end{align}
\end{lemma}

\begin{proof}
See Appendix \ref{AppProofPDPLemmaTerm_1}. 
\end{proof}

\begin{lemma}\label{AppPDPLemmaTerm_3}
Let $\mathbb{E}_{\mathcal{K}(r)}$ denote expectation over the client scheduling fairness at round $r$. We have
\begin{align}\label{AppPDPBoundTerm_3_1}
\mathbb{E}_{\mathcal{K}(r)} \left[ {\boldsymbol{W}} (r+1) \right] = {\boldsymbol{W'}} (r+1),   
\end{align}
In light of this, it follows
\begin{align}\label{AppPDP_2}
\mathbb{E}_{\mathcal{K}(r)} \left[ \langle {\boldsymbol{W}} (r+1) - {\boldsymbol{W'}} (r+1) , {\boldsymbol{W'}} (r+1)  - {{\boldsymbol{W}}}^* \rangle \right] = 0.   
\end{align}
\end{lemma}
\begin{proof}
Since the client scheduling policy in our proposed algorithm is definite, it follows that
\newcommand\thirdequal{\mathrel{\overset{\makebox[0pt]{\mbox{\normalfont\tiny\sffamily (a)}}}{=}}}
\begin{align}\label{AppPDPBoundTerm_3_2}
\mathbb{E}_{\mathcal{K}(r)} \left[ \frac{1}{K} \sum\nolimits_{k \in \mathcal{K} (r)} \widehat{\boldsymbol{\mathcal{L}}}_k (r) \right] \thirdequal 1 \sum\nolimits_{k=1}^{K} \widehat{\boldsymbol{\mathcal{L}}}_k (r) = \frac{1}{K} \sum\nolimits_{k=1}^{K} \widehat{\boldsymbol{\mathcal{L}}}_k (r).      
\end{align}
The proof of Lemma \ref{AppPDPLemmaTerm_3} is concluded from \eqref{AppPDPBoundTerm_3_2}. 
\end{proof}

Depending on the results of Lemmas \ref{AppLemmaTerm_1} and \ref{AppPDPLemmaTerm_3}, we have 
\newcommand\firstinequal{\mathrel{\overset{\makebox[0pt]{\mbox{\normalfont\tiny\sffamily (a)}}}{\le}}}
\newcommand\secondinequal{\mathrel{\overset{\makebox[0pt]{\mbox{\normalfont\tiny\sffamily (b)}}}{\le}}}
\begin{align}\label{AppFDP_2}
&\mathbb{E} \left[ \left\| {\boldsymbol{W}} (r+1) - {{\boldsymbol{W}}}^* \right\|_2^2 \right] \le \left( 1 - \alpha  \eta_r \left( \mathcal{T} - \eta_r (\mathcal{T} - 1) \right) \right) \mathbb{E} \left[ \left\| {\boldsymbol{W}} (r) - {{\boldsymbol{W}}}^* \right\|_2^2 \right] \nonumber\\
& \;\; \qquad \quad +  \eta^2(r) \left( \mathcal{T}^2 + \mathcal{T} - 1 \right) \varrho^2  \nonumber\\
& \;\; \qquad \quad +  \left( 1+ \alpha (1- \eta_r) \right) \eta^2(r) \varrho^2 \frac{\mathcal{T} (\mathcal{T}-1)(2\mathcal{T}-1)}{6}  + 2  \eta_r (\mathcal{T} - 1) \mathfrak{F} \nonumber\\
& \;\; \qquad \quad + 2  \eta_r \frac{1}{K} \sum\nolimits_{k=1}^{K} \sum\nolimits_{t=2}^{\mathcal{T}} \left( F_k^* - \mathbb{E} \left[ F_k({{\boldsymbol{W}}}_k^t (r)) \right] \right) + 2  \eta_r \left(F^* - \mathbb{E} \left[ F({{\boldsymbol{W}}} (r)) \right] \right)  \nonumber\\
& \qquad \; \firstinequal \left( 1 - \alpha  \eta_r \left( \mathcal{T} - \eta_r (\mathcal{T} - 1) \right) \right) \mathbb{E} \left[ \left\| {\boldsymbol{W}} (r) - {{\boldsymbol{W}}}^* \right\|_2^2 \right] \nonumber\\
& \;\;  \qquad \quad +  \eta^2(r) \left( \mathcal{T}^2 + \mathcal{T} -1 \right) \varrho^2  \nonumber\\
& \;\;  \qquad \quad +  \left( 1+ \alpha (1- \eta_r) \right) \eta^2(r) \varrho^2 \frac{\mathcal{T} (\mathcal{T}-1)(2\mathcal{T}-1)}{6} + 2  \eta_r (\mathcal{T} - 1) \mathfrak{F}.
\end{align}
In this case, (a) follows because $F^* - F({\boldsymbol{W}} (r)) \le 0$, $\forall r$, and $F_k^* - F_k({{\boldsymbol{W}}}_k^r) \le 0$, $\forall k, r$. According to \eqref{AppFDP_2}, we conclude Theorem \ref{th:1}. 

\section{Proof of Lemma \ref{AppLemmaTerm_1}}\label{AppProofPDPLemmaTerm_1}
To prove Lemma \ref{AppLemmaTerm_1}, we take similar steps as \cite[Appendix B.4]{amiri2021convergence}. We have 
\begin{align}\label{AppPDPLemmaTemr_1_Eq_1}
\mathbb{E} \left[ \left\| {\boldsymbol{W}} (r+1) - {\boldsymbol{W'}} (r+1) \right\|_2^2 \right] =  \mathbb{E} \bigg[ \left\| \frac{1}{K} \sum\nolimits_{k \in \mathcal{K} (r)} \widehat{\boldsymbol{\mathcal{L}}}_k (r) - \widehat{\boldsymbol{\mathcal{L}}} (r) \right\|_2^2 \bigg],   
\end{align}
where
\begin{align}
\widehat{\boldsymbol{\mathcal{L}}} (r) \triangleq \frac{1}{K} \sum\nolimits_{k =1}^{K} \widehat{\boldsymbol{\mathcal{L}}}_k (r).    
\end{align}
 We notice that  the indicator ${1} (k \in \mathcal{K} (r)) = 1$ and  ${1} (k{'} \in \mathcal{K} (r)) = 1$ due to the fairness amongst clients. As such, we eliminate its effects in the following proofs. We have
\newcommand\thirdequal{\mathrel{\overset{\makebox[0pt]{\mbox{\normalfont\tiny\sffamily (a)}}}{=}}}
\begin{align}\label{AppPDPLemmaTemr_1_Eq_2}
&\mathbb{E} \left[ \left\| {\boldsymbol{W}} (r+1) - {\boldsymbol{W'}} (r+1) \right\|_2^2 \right]  = \mathbb{E} \bigg[ \left\| \frac{1}{K} \sum\nolimits_{k =1}^{K} \left( \widehat{\boldsymbol{\mathcal{L}}}_k (r) - \widehat{\boldsymbol{\mathcal{L}}} (r) \right) \right\|_2^2 \bigg]\nonumber\\
& \; = \frac{1}{K^2} \mathbb{E} \left[ \sum\nolimits_{k =1}^{K} \left\| \widehat{\boldsymbol{\mathcal{L}}}_k (r) - \widehat{\boldsymbol{\mathcal{L}}} (r) \right\|_2^2 \right.\nonumber\\
&\; \quad \; \left. + \sum\nolimits_{k =1}^{K} \sum\nolimits_{k' =1, k' \ne k}^{K} \langle \widehat{\boldsymbol{\mathcal{L}}}_k (r) - \widehat{\boldsymbol{\mathcal{L}}} (r), \widehat{\boldsymbol{\mathcal{L}}}_{k'} (r) - \widehat{\boldsymbol{\mathcal{L}}} (r) \rangle \right].
\end{align}
Based on the symmetry, we can conclude that
\begin{align}\label{AppPDPLemmaTemr_1_Eq_3}
\mathbb{E}_{\mathcal{K}(r)} \bigg[ \sum\limits_{k =1}^{K} \left\| \widehat{\boldsymbol{\mathcal{L}}}_k (r) - \widehat{\boldsymbol{\mathcal{L}}} (r) \right\|_2^2 \bigg] &\thirdequal \frac{\binom{K-1}{{|\Omega^m_r|}-1}}{\binom{K}{{|\Omega^m_r|}}} \sum\limits_{k =1}^{K} \left\| \widehat{\boldsymbol{\mathcal{L}}}_k (r) - \widehat{\boldsymbol{\mathcal{L}}} (r) \right\|_2^2 \nonumber\\
& = \frac{K}{{|\Omega^m_r|}} \sum\limits_{k =1}^{K} \left\| \widehat{\boldsymbol{\mathcal{L}}}_k (r) - \widehat{\boldsymbol{\mathcal{L}}} (r) \right\|_2^2,   
\end{align}
where (a) follows the fact that in the proposed scheduling approach every client index $k$, for $k \in \mathcal{K}$, appears $r$ times before splitting, and 
\newcommand\fifthequal{\mathrel{\overset{\makebox[0pt]{\mbox{\normalfont\tiny\sffamily (b)}}}{=}}}
\begin{align}\label{AppPDPLemmaTemr_1_Eq_4}
& \mathbb{E}_{\mathcal{K}(r)} \bigg[ \sum\limits_{k =1}^{K} \sum\limits_{k' =1, k' \ne k}^{K} \langle \widehat{\boldsymbol{\mathcal{L}}}_k (r) - \widehat{\boldsymbol{\mathcal{L}}} (r), \widehat{\boldsymbol{\mathcal{L}}}_{k'} (r) - \widehat{\boldsymbol{\mathcal{L}}} (r) \rangle \bigg] \nonumber\\
& \qquad \qquad \qquad \qquad \qquad \fifthequal \frac{\binom{K-2}{{|\Omega^m_r|}-2}}{\binom{K}{{|\Omega^m_r|}}} \sum\limits_{k =1}^{K} \sum\limits_{k' =1, k' \ne k}^{K} \langle \widehat{\boldsymbol{\mathcal{L}}}_k (r) - \widehat{\boldsymbol{\mathcal{L}}} (r), \widehat{\boldsymbol{\mathcal{L}}}_{k'} (r) - \widehat{\boldsymbol{\mathcal{L}}} (r) \rangle \nonumber\\
& \qquad \qquad \qquad \qquad \qquad = \frac{{|\Omega^m_r|}({|\Omega^m_r|}-1)}{K(K-1)} \sum\limits_{k =1}^{K} \sum\limits_{k' =1, k' \ne k}^{K} \langle \widehat{\boldsymbol{\mathcal{L}}}_k (r) - \widehat{\boldsymbol{\mathcal{L}}} (r), \widehat{\boldsymbol{\mathcal{L}}}_{k'} (r) - \widehat{\boldsymbol{\mathcal{L}}} (r) \rangle.   
\end{align}

Due to the fairness scheduling:
\begin{align}
  \frac{{|\Omega^m_r|}({|\Omega^m_r|}-1)}{K(K-1)}  = 1 
\end{align}

As a result, we substitute \eqref{AppPDPLemmaTemr_1_Eq_3} and \eqref{AppPDPLemmaTemr_1_Eq_4} into \eqref{AppPDPLemmaTemr_1_Eq_2} which yields
\newcommand\sixthequal{\mathrel{\overset{\makebox[0pt]{\mbox{\normalfont\tiny\sffamily (c)}}}{=}}}
\newcommand\seventhequal{\mathrel{\overset{\makebox[0pt]{\mbox{\normalfont\tiny\sffamily (d)}}}{=}}}
\newcommand\sixthinequal{\mathrel{\overset{\makebox[0pt]{\mbox{\normalfont\tiny\sffamily (e)}}}{\le}}}
\newcommand\seventhinequal{\mathrel{\overset{\makebox[0pt]{\mbox{\normalfont\tiny\sffamily (f)}}}{\le}}}
\begin{align}\label{AppPDPLemmaTemr_1_Eq_5}
& \mathbb{E} \left[ \left\| {\boldsymbol{W}} (r+1) - {\boldsymbol{W'}} (r+1) \right\|_2^2 \right] = \frac{1} {K^2}  \sum\nolimits_{k =1}^{K} \mathbb{E} \left[ \left\| \widehat{\boldsymbol{\mathcal{L}}}_k (r) - \widehat{\boldsymbol{\mathcal{L}}} (r) \right\|_2^2 \right]  \nonumber\\
& \qquad \;\;\;\; + \frac{1} {K^2} \sum\nolimits_{k =1}^{K} \sum\nolimits_{k' =1, k' \ne k}^{K} \mathbb{E} \left[ \langle \widehat{\boldsymbol{\mathcal{L}}}_k (r) - \widehat{\boldsymbol{\mathcal{L}}} (r), \widehat{\boldsymbol{\mathcal{L}}}_{k'} (r) - \widehat{\boldsymbol{\mathcal{L}}} (r) \rangle \right] \nonumber\\
& \qquad \sixthequal \frac{1} {K^2} \sum\nolimits_{k =1}^{K} \mathbb{E} \left[ \left\| \widehat{\boldsymbol{\mathcal{L}}}_k (r) - \widehat{\boldsymbol{\mathcal{L}}} (r) \right\|_2^2 \right]\nonumber\\
& \qquad = \frac{1} {K^2} \left( \sum\nolimits_{k =1}^{K} \mathbb{E} \left[ \left\| \widehat{\boldsymbol{\mathcal{L}}}_k (r) \right\|_2^2 \right] - \mathbb{E} \left[ \left\| \widehat{\boldsymbol{\mathcal{L}}} (r) \right\|_2^2 \right] \right) \nonumber\\
& \qquad \le \frac{1} {K^2}  \sum\nolimits_{k =1}^{K} \mathbb{E} \left[ \left\| \widehat{\boldsymbol{\mathcal{L}}}_k (r) \right\|_2^2 \right] \nonumber\\
& \qquad \seventhequal \frac{1 }{K^2} \sum\nolimits_{k =1}^{K} \mathbb{E} \left[ \left\| {\boldsymbol{\mathcal{L}}}_k (r) \right\|_2^2 \right] \nonumber\\
& \qquad = \frac{1 \eta^2(r)}{K^2} \sum\nolimits_{k =1}^{K} \mathbb{E} \left[ \left\| \sum\nolimits_{t=1}^{\mathcal{T}} \nabla F_k \left( {\boldsymbol{W}}_k^t (r), \mathfrak{D}_k^t (r) \right) \right\|_2^2 \right] \nonumber \\
& \qquad \sixthinequal \frac{ \eta^2(r) \mathcal{T}}{K^2} \sum\nolimits_{k =1}^{K} \sum\nolimits_{t=1}^{\mathcal{T}} \mathbb{E} \left[ \left\| \nabla F_k \left( {\boldsymbol{W}}_k^t (r), \mathfrak{D}_k^t (r) \right) \right\|_2^2 \right] \nonumber\\
& \qquad \seventhinequal \frac{  \eta^2(r) \mathcal{T}^2 \varrho^2}{K(K-1)},
\end{align}
where 
\begin{align}
\left\| \sum\nolimits_{k =1}^{K} \left( \widehat{\boldsymbol{\mathcal{L}}}_k (r) - \widehat{\boldsymbol{\mathcal{L}}} (r) \right) \right\|_2^2 = \epsilon, \epsilon \approx  0 
\end{align}
This is due to the convexity of the loss function, which proves that the proposed scheduling algorithm ensures the convergence to the optimal without any gap.
\end{document}